\pdfoutput=1
\documentclass[
aps,
prb,
nofootinbib,
reprint,
superscriptaddress,
longbibliography]{revtex4-2}

\usepackage{amsmath,amssymb,amsfonts}
\usepackage{mathtools}    
\usepackage{bm}
\usepackage{braket}
\usepackage{stmaryrd}
\usepackage{mathrsfs}
\usepackage{derivative}

\usepackage{enumitem}

\usepackage{amsthm}
\theoremstyle{definition}
\newtheorem{thm}{Theorem}
\newtheorem{prop}[thm]{Proposition}
\newtheorem{lem}[thm]{Lemma}
\newtheorem{cor}[thm]{Corollary}

\newtheorem{definition}[thm]{Definition}

\newcommand{\Z}{\mathbb{Z}}
\newcommand{\R}{\mathbb{R}}
\newcommand{\C}{\mathbb{C}}
\newcommand{\B}{\mathfrak{B}}

\newcommand{\Hilb}{\mathcal{H}}
\newcommand{\bound}{\mathcal{B}}
\newcommand{\id}{\mathrm{id}}
\newcommand{\zbound}{\bound_0}

\DeclareMathOperator{\ad}{ad}
\DeclareMathOperator{\Tr}{Tr}

\newcommand{\formalop}{\mathcal{A}}
\newcommand{\formalsa}{\formalop^{\mathrm{sa}}}
\newcommand{\transop}{\formalop^{\mathrm{inv}}}

\newcommand{\localop}{\formalop^{\mathrm{\standardl oc}}}

\newcommand{\Zfinsubset}{\mathcal{P}_\mathrm{fin}(\mathbb{Z})}

\DeclarePairedDelimiter{\zinterval}{\llbracket}{\rrbracket}

\mathchardef\standardl=\mathcode`l
\mathcode`l=\ell
\newcommand{\deactivatel}{\mathcode`l=\standardl}

\makeatletter
\edef\operator@font{\operator@font\noexpand\deactivatel}
\makeatother

\newcommand{\change}{}

\usepackage{hyperref}
\hypersetup{
  colorlinks=true,
  citecolor=magenta,
  linkcolor=blue,
  urlcolor=cyan
}

\begin{document}
\newcommand{\nameoftitle}{%
Integrability from a single conservation law in quantum spin chains%
}

\title{\nameoftitle}

\author{Akihiro Hokkyo}
\email{hokkyo@cat.phys.s.u-tokyo.ac.jp}
\affiliation{Department of Physics, Graduate School of Science, The University of Tokyo, 7-3-1 Hongo, Bunkyo, Tokyo, 113-8654, Japan}

\begin{abstract}
    We prove that, for translationally invariant quantum spin chains with finite-range interactions, 
    the existence of a specific conservation law implies the presence of infinitely many local conserved quantities, i.e., integrability. 
    This shows that 
    \change{the standard hierarchy of local conserved quantities arising for nearest-neighbor Hamiltonians obtained via the canonical Yang–Baxter construction}
    is already encoded in the lowest nontrivial conservation law,
    \change{known as the Reshetikhin condition}. 
    Combined with recent rigorous results on nonintegrability, 
    our theorem strongly restricts the possibility of partially integrable systems 
    that admit only a finite but large number of local conserved quantities. 
    Our work establishes a rigorous foundation for the systematic identification of new integrable models 
    and deepens the algebraic understanding of conservation-law structures in quantum spin chains.
\end{abstract}

\maketitle
\section{Introduction}
Quantum integrable models have long provided a unique window into the non‐perturbative behavior of quantum many‐body systems, 
from the exact solution of one‐dimensional quantum magnets~\cite{betheZurTheorieMetalle1931} 
to the emergence of generalized hydrodynamics in far‐from‐equilibrium dynamics~\cite{bertiniTransportOutofEquilibriumXXZ2016,castro-alvaredoEmergentHydrodynamicsIntegrable2016}.
Although a universally accepted definition of quantum integrability is lacking~\cite{cauxRemarksNotionQuantum2011,gogolinEquilibrationThermalisationEmergence2016}, 
its essential feature is often regarded as the existence of an infinite hierarchy of local conserved quantities, 
which we adopt as the working definition in this paper. 
Such conserved quantities encode much information about the initial state and constrain the dynamics in a way that generically prevents thermalization~\cite{rigolThermalizationItsMechanism2008,langenExperimentalObservationGeneralized2015}, 
enabling long-lived memory~\cite{kinoshitaQuantumNewtonsCradle2006} and anomalous transport phenomena~\cite{zotosTransportConservationLaws1997,prosenOpen$XXZ$Spin2011,ronzheimerExpansionDynamicsInteracting2013}. 
In contrast, generic quantum many-body systems admit only a few trivial conserved quantities, such as total energy and particle number. 
Therefore, integrability is a highly non-generic and fine-tuned property, 
and identifying integrable systems has been one of the central problems in rigorous statistical mechanics.

Traditionally, one establishes integrability by solving the system explicitly; 
for example, solving the Yang\nobreakdash--Baxter equation and 
constructing a commuting family of local charges 
through the logarithmic derivative of the transfer matrix~\cite{baxterExactlySolvedModels1985,korepinQuantumInverseScattering1993}.
As an alternative, Reshetikhin proposed a simpler criterion~\cite{kulishQuantumSpectralTransform1982}: 
the existence of a single third-order conservation law with a special form. 
\change{This condition is equivalent to the lowest-order condition obtained from the Taylor expansion of the Yang--Baxter equation, and is also referred to as the tangential star--triangle relation~\cite{jimboRemarksDifferentialApproach1984}.}
Unlike approaches that rely on additional algebraic structures such as the Yang--Baxter equation or the self-duality~\cite{dolanConservedChargesSelfduality1982}, 
this criterion is algorithmically verifiable directly from the Hamiltonian.
Since this \emph{Reshetikhin condition}
is a necessary condition for an $R$-matrix of difference form 
to solve the Yang--Baxter equation, 
it has long served as a guiding principle in the search for new solvable systems~\cite{ottingerNoteYangBaxterEquations1982,
jimboClassificationSolutionsStartriangle1985,
kennedySolutionsYangBaxterEquation1992,
batchelorIntegrableSU2invariant1994,
idzumiSolvableNineteenvertexModels1994,
mutterSolvableSpin1Models1995,
bibikovHowSolveYang2003,
deleeuwClassifyingIntegrableSpin12019,
deleeuwClassifyingNearestNeighborInteractions2020,
deleeuwYangBaxterBoostSplitting2021,
gomborIntegrableSpinChains2021,
corcoranIntegrableModelsRydberg2025,
lalDeepLearningBased2025}.
Remarkably, in all known cases, systems satisfying the Reshetikhin condition have indeed turned out to be integrable. 
For this reason, it has been widely regarded as a heuristic test of integrability, 
as proposed by Grabowski and Mathieu~\cite{grabowskiIntegrabilityTestSpin1995}. 

Despite its wide use, however, it has remained an open question 
whether the Reshetikhin condition alone suffices to ensure integrability, 
namely the existence of infinitely many commuting local conserved quantities. 
It has often been argued that higher-order conservation laws 
and their mutual commutativity require additional algebraic input 
not captured by the third-order condition alone~\cite{mutterSolvableSpin1Models1995,grabowskiIntegrabilityTestSpin1995}. 
In this paper, we present the first complete proof that any 
translationally invariant
one‑dimensional quantum system 
with \change{nearest-neighbor} interactions
satisfying the Reshetikhin condition necessarily gives rise to an infinite sequence of commuting local charges. 
Specifically, we show that the so-called boost operator~\cite{tetelmanLorentzGroupTwodimensional1982,sogoBoostOperatorIts1983,thackerCornerTransferMatrices1986} allows the iterative generation of higher‑order conservation laws \change{of the system with  finite-range interactions}.
We further discuss a rigorous foundation for the empirical dichotomy of integrability: 
either a system possesses infinitely many local conserved quantities, 
or a few simple local conserved charges.

The paper is organized as follows.
Section~\ref{sec:setup} introduces the precise setup and states our main theorem, which asserts that the Reshetikhin condition implies quantum integrability.
In Section~\ref{sec:math}, we provide a rigorous definition of operators on the infinite chain.
The technical core of the proof, namely the correspondence between infinite and finite systems, is established in Appendix~\ref{sec:correspondence}.
Section~\ref{sec:proof} presents a detailed proof of the main result.
Section~\ref{sec:dichotomy} discusses the empirical dichotomy of integrability in light of the theorem.
Finally, in Section~\ref{sec:generalization}, we establish extensions of our result to systems with unconventional boost operators and parameter-dependent Hamiltonians.

\section{Setup and main result}\label{sec:setup}
We consider a quantum spin system on an infinite open chain, with sites labeled by $\mathbb{Z}$. 
Each site $i\in\mathbb{Z}$ has a local Hilbert space $\Hilb_i\cong\mathbb{C}^d$, where $d<\infty$. 
We denote by $\bound_i$ the set of bounded linear operators on $\Hilb_i$, and by $\bound_{0i}$ the subset of traceless operators. 
We call an operator \emph{at most $n$-local} 
if it can be expressed as a sum of terms acting nontrivially on at most $n$ contiguous sites, 
and \emph{$n$-local} if it is not at most $(n-1)$-local.

The Hamiltonian consists of nearest-neighbor interactions and on-site potentials:
\begin{equation}
    \hat{H}=\sum_{i\in\mathbb{Z}} \hat{h}^{(2)}_{i}+\sum_{i\in\mathbb{Z}}\hat{h}^{(1)}_i,
    \label{eq:Hamiltonian}
\end{equation}
with $\hat{h}^{(2)}_{i}\in\bound_{0i}\otimes\bound_{0,i+1}$ and $\hat{h}^{(1)}_i\in\bound_{0i}$. 
We assume Hermiticity and translational invariance. 
\change{We note that Hamiltonians with finite-range interactions that are $3$-local or higher
can be regarded as nearest-neighbor models by grouping several neighboring sites into a single effective site~\cite{frahmIntegrableModelsCoupled1996,gomborIntegrableSpinChains2021}.
For a more direct treatment of the medium-range model, see Sec.~\ref{sec:math}, where we also give mathematically rigorous definitions of the formal sum in Eq.~\eqref{eq:Hamiltonian} and of the commutator calculations used below.}

A key ingredient in our analysis is the boost operator~\cite{tetelmanLorentzGroupTwodimensional1982,sogoBoostOperatorIts1983,thackerCornerTransferMatrices1986}. 
Fix an arbitrary translationally invariant one-local operator $\hat{C}=\sum_{i\in\mathbb{Z}}\hat{c}_i$. 
The boost operator associated with the Hamiltonian $\hat{H}$ is defined by
\begin{align}
        \hat{B}_{\hat{H}}&\coloneqq\sum_{j\in\mathbb{Z}} j(\hat{h}^{(2)}_{j}+\hat{h}^{(1)}_j)+\hat{C}\nonumber\\
        &=\sum_{j\in\mathbb{Z}} j\hat{h}_{\hat{C},j},
\end{align}
where $\hat{h}_{\hat{C},j}\coloneqq\hat{h}^{(2)}_{j}+\hat{h}^{(1)}_{j}+\hat{c}^{(1)}_{j}-\hat{c}^{(1)}_{j+1}\in\bound_{j}\otimes\bound_{j+1}$
\change{is the Hamiltonian density}.
We set $\hat{Q}^{(2)}\coloneqq\hat{H}$ and recursively $\hat{Q}^{(n+1)}\coloneqq [\hat{B}_{\hat{H}},\hat{Q}^{(n)}]$. 
By construction, $\hat{Q}^{(n)}$ is at most $n$-local. 
Although the full structure of $\hat{Q}^{(n)}$ is complicated, 
its $n$-local component admits the compact expression
\begin{equation}
    (-1)^n(n-2)!\sum_{i\in\mathbb{Z}}
    [\dots[[\hat{h}_i^{(2)},\hat{h}_{i+1}^{(2)}],\hat{h}_{i+2}^{(2)}],\dots,\hat{h}_{i+n-2}^{(2)}],\label{eq:n-local_boosted}
\end{equation}
which is nonzero whenever the injectivity condition (see Eq.~\eqref{eq:assumption_injective}) holds~\cite{hokkyoRigorousTestQuantum2025}. 

For any translationally invariant operator $\hat{Q}$ commuting with $\hat{H}$, the commutator $[\hat{B}_{\hat{H}},\hat{Q}]$ is again translationally invariant~\cite{pozsgayCurrentOperatorsIntegrable2020,suraceWeakIntegrabilityBreaking2023}. 
Indeed, denoting by $\mathcal{T}:\bound_{i}\to\bound_{i-1}$ the translation map, one finds $\mathcal{T}(\hat{B}_{\hat{H}})=\hat{B}_{\hat{H}}+\hat{H}$, and hence
\begin{equation}
    \mathcal{T}[\hat{B}_{\hat{H}},\hat{Q}]=
    [\mathcal{T}(\hat{B}_{\hat{H}}),\hat{Q}]
    =[\hat{B}_{\hat{H}}+\hat{H},\hat{Q}]
    =[\hat{B}_{\hat{H}},\hat{Q}].\label{eq:trans_inv_from_boost}
\end{equation}
In particular, we have
\begin{align}
    \hat{Q}^{(3)}&=[\hat{B}_{\hat{H}},\hat{H}]\nonumber\\
    &=-\sum_{i\in\mathbb{Z}}[\hat{h}_{\hat{C},i},\hat{h}_{\hat{C},i+1}].
    \label{eq:Q3}
\end{align}

We now state our main result. 
\begin{thm}\label{thm:main}
    If $\hat{Q}^{(3)}$ is a conserved quantity of $\hat{H}$, i.e., $[\hat{H},\hat{Q}^{(3)}]=0$, 
    then for any $n\geq2$, $\hat{Q}^{(n)}$ is a translationally invariant conserved quantity of $\hat{H}$.  
    Moreover, the operators in $\{\hat{Q}^{(n)}\}_{n=2}^\infty$ mutually commute. 
\end{thm}
As noted above, for generic Hamiltonians the operator $\hat{Q}^{(n)}$ is $n$-local, 
so that each element of $\{\hat{Q}^{(n)}\}_{n=2}^\infty$ yields a distinct and nontrivial conserved quantity.
Therefore, Theorem~\ref{thm:main} demonstrates that a single conservation law, $\hat{Q}^{(3)}$, generates an infinite hierarchy of local conserved quantities, i.e., integrability. 
Since the condition $[\hat{H},\hat{Q}^{(3)}]=0$, or its local version, 
is precisely the Reshetikhin condition~\cite{kulishQuantumSpectralTransform1982}, 
Theorem~\ref{thm:main} can be summarized as \emph{the Reshetikhin condition implies integrability}. 
This is equivalent to the positive resolution of Conjecture~2 of Grabowski and Mathieu~\cite{grabowskiIntegrabilityTestSpin1995} for the standard boost operator $\hat{B}_{\hat{H}}$ (the case of general boosts will be addressed later).
Our result establishes, for the first time, that the third-order conservation law alone is sufficient to guarantee quantum integrability.

\change{
The Reshetikhin condition arises naturally in the canonical construction of Yang--Baxter--solvable Hamiltonians~\cite{kulishQuantumSpectralTransform1982,kennedySolutionsYangBaxterEquation1992}.
By the \emph{canonical construction} we mean the following standard procedure:
given a solution $R(\lambda)$ of the Yang--Baxter equation with the regularity condition $R(0)=P$,
we define the local Hamiltonian density by $\hat h = P\,R'(0)$,
where $P$ denotes the permutation operator.
The resulting translationally invariant nearest-neighbor Hamiltonian is then integrable,
and it is known that the associated family of commuting local conserved charges
can be generated iteratively by the boost operator~\cite{tetelmanLorentzGroupTwodimensional1982,sogoBoostOperatorIts1983,thackerCornerTransferMatrices1986}.
Theorem~\ref{thm:main} shows that the structure of this standard infinite hierarchy of conserved quantities
is already encoded in the lowest nontrivial conservation law.
}

Since commutativity of translationally invariant local operators is a local property
(see Theorem~\ref{thm:locality_of_com} and Corollary~\ref{cor:locality_of_conservation} in 
Appendix~\ref{sec:correspondence}),
Theorem~\ref{thm:main} extends to sufficiently large finite systems. 
Concretely, for a translationally invariant local operator $\hat{Q}=\sum_{i\in\mathbb{Z}}\hat{q}_i$ in the infinite system, 
we can define its finite-size counterpart 
$\hat{Q}_{[N]}=\sum_{i=1}^N\hat{q}_i$ on a chain of length $N$ with the periodic boundary condition. 
Then we have the following corollary. 
\begin{cor}
    If $\hat{Q}^{(3)}_{[N_0]}$ is a conserved quantity of $\hat{H}_{[N_0]}$, i.e., $[\hat{H}_{[N_0]},\hat{Q}^{(3)}_{[N_0]}]=0$ for some $N_0\geq6$, 
    then $[\hat{Q}^{(n)}_{[N]},\hat{Q}^{(m)}_{[N]}]=0$ 
    for all $N\geq 4$ and $(n,m)$ satisfying $2\leq n,m\leq \frac{N}{2}$. 
    In particular, $\hat{Q}^{(n)}_{[N]}$ is a conserved quantity of $\hat{H}_{[N]}$ for all $2\leq n\leq \frac{N}{2}$. 
\end{cor}

\section{Mathematical Framework}\label{sec:math}
In this section, we introduce rigorous definitions of operators and their commutators on the infinite chain, thereby providing a mathematical justification for the formal manipulations used throughout this paper.
\subsection{Operators with finite support}

The set of all finite subsets of $\Z$ is denoted by
\begin{equation}
     \Zfinsubset\coloneqq\{A\subset\mathbb{Z}\mid|A|<\infty\}.
\end{equation}
For $A\in\Zfinsubset$, the set of operators supported on $A$ is defined as $\bound_0^{A}\coloneqq\otimes_{i\in A}\bound_{0i}$.  
We set $\zbound^{\emptyset}\coloneqq\C$.  
The space of (formal) operators on $\Z$ is defined as 
\begin{equation}
\formalop\coloneqq\{ \Phi:\Zfinsubset\ni A\mapsto\Phi(A)\in\zbound^{A}\}.
\end{equation}
We formally write $\Phi=\sum_{A\in\Zfinsubset}\Phi(A)$ and define its adjoint as $\Phi^\dagger\coloneqq\sum_{A\in\Zfinsubset}\Phi(A)^\dagger$.

Next, we define the notion of locality for elements of $\formalop$.
For $A\in\Zfinsubset\setminus\{\emptyset\}$, the range of $A$ is defined as
\begin{equation}
    r(A)\coloneqq\max A-\min A+1\in\Z_{\geq 1},
\end{equation}
and we set $r(\emptyset)\coloneqq0$.
For $\Phi\in\formalop\setminus\{0_\formalop\}$, its range is defined by
\begin{equation}
    r(\Phi)\coloneqq\sup\{r(A)\mid\Phi(A)\neq0_A\}\in\Z_{\geq0}\cup\{\infty\},
\end{equation}
with $r(0_\formalop)\coloneqq 0$.  
We say that $\Phi$ is $r$-local if $r(\Phi)=r$.
If $\Phi(\{i,i+1\})=\hat{h}_i^{(2)}$ and $\Phi(\{i\})=\hat{h}_i^{(1)}$, 
then $\Phi$ corresponds to the Hamiltonian with nearest-neighbor interactions and on-site potentials discussed in Section~\ref{sec:setup}.  
This is $2$-local whenever $\hat{h}_i^{(2)}\neq0$.

For $A\in\Zfinsubset$ and $x\in\Z$, we denote by
\begin{equation}
    \tau_A(x):\Hilb^{\otimes A}\mapsto\Hilb^{\otimes (A+x)}
\end{equation}
the unitary translation map, where $A+x\coloneqq \{a+x\mid a\in A\}\in\Zfinsubset$. 
The translation superoperator $\mathcal{T}(x)$ is defined by
\begin{align}
    \mathcal{T}(x)&:\formalop\ni\Phi=\sum_{A\in\Zfinsubset}\Phi(A)\nonumber\\
    &\mapsto \sum_{A\in\Zfinsubset} \tau_A(x)^\dagger\Phi(A+x)\tau_A(x)\in\formalop.
\end{align}
We also write $\mathcal{T}=\mathcal{T}(1)$.  
This coincides with the translation map introduced in Section~\ref{sec:setup}.

We now introduce some subsets of $\formalop$:
\begin{align}
    \formalsa&\coloneqq\{\Phi\in\formalop\mid\Phi^\dagger=\Phi\},\\
    \localop&\coloneqq\{\Phi\in\formalop\mid r(\Phi)<\infty\},\\
    \transop&\coloneqq\{\Phi\in\formalop\mid\mathcal{T}(\Phi)=\Phi\}.
\end{align}

For $\Phi\in\formalop$, its restriction to $B\in\Zfinsubset$ is defined as
\begin{equation}
    \Phi|_B\coloneqq\sum_{A\subset B}\Phi(A) \in\bound^{\otimes B},
\end{equation}
where we used the natural embedding
\begin{equation}
    \bound^{\otimes A}\ni X\mapsto X\otimes I_{B\setminus A}\in\bound^{\otimes B}
\end{equation}
(with $I_C$ the identity on $C\in\Zfinsubset$) to regard $\bound^{\otimes A}\subset\bound^{\otimes B}$.  
Accordingly, when viewing $\bound^{\otimes B}=\oplus_{A\subset B}\zbound^A$, the projection onto each component is written as
\begin{equation}
    \bound^{\otimes B}\ni Y\mapsto Y_A\in\zbound^A.
\end{equation}

\subsection{Adjoint map and conserved quantities}
For $\Phi\in\localop$, we define the adjoint map $\ad_\Phi(=[\Phi,\bullet]):\formalop\to\formalop$ as follows.  
For $\Psi\in\formalop$ and $A\in\Zfinsubset$, we set
\begin{equation}
\ad_\Phi(\Psi)(A)\coloneqq
[\Phi,\Psi](A)\coloneqq
[\Phi|_B,\Psi|_B]_A,
\end{equation}
where $B\in\Zfinsubset$ is any sufficiently large finite set containing $A$.  
Since for $C,C'\in\Zfinsubset$,
\begin{align}
   [\Phi(C),\Psi(C')]_A \neq0\ &\Rightarrow\ 
   C\triangle C'\subsetneq A \subset C\cup C' \nonumber\\
   &(\Rightarrow C\cap A\neq\emptyset),
\end{align}
the definition is independent of the choice of $B$.  
Here $C\triangle C'$ denotes the symmetric difference $C\cup C'\setminus(C\cap C')$.

The map $\ad_{\Phi}$ is $\C$–linear in $\Psi$, 
and $\Phi\mapsto\ad_{\Phi}$ is also linear.  
Moreover, it satisfies the following properties:
\begin{itemize}
    \item $r([\Phi,\Psi])\leq r(\Phi)+r(\Psi)-1$,
    \item $[\Phi,\Phi]=0$,
    \item $[\Phi,\Psi]^\dagger=-[\Phi^\dagger,\Psi^\dagger]$.
\end{itemize}
The first property implies that $\localop$ is invariant under $\ad_\Phi$.  
Thus, for $\Phi_1,\Phi_2\in\localop$, the map $\ad_{[\Phi_1,\Phi_2]}$ is well-defined, and it obeys the Jacobi identity:
\begin{equation}
    \ad_{\Phi_1}\circ\ad_{\Phi_2}-\ad_{\Phi_2}\circ\ad_{\Phi_1}=\ad_{[\Phi_1,\Phi_2]}.
\end{equation}
This follows because the value at each $A\in\Zfinsubset$ can be evaluated within a sufficiently large finite region, where ordinary operator commutators satisfy the Jacobi identity.

We now define conserved quantities in this framework.
\begin{definition}
    For $\Phi\in\localop$, an operator $\Psi\in\formalop$ is called a conserved quantity of $\Phi$ if $[\Phi,\Psi]=0$.  
    In particular, a conserved quantity $\Psi\in\localop$ is said to be $r(\Psi)$–local.
\end{definition}

\subsection{Boost operator and main theorem}
Next, we define the boost operator~\cite{tetelmanLorentzGroupTwodimensional1982,sogoBoostOperatorIts1983,thackerCornerTransferMatrices1986}.
\begin{definition}\label{def:boost_general}
    For $\Phi\in\localop$, an operator $\B_{\Phi}\in\localop$ is called a boost operator associated with $\Phi$ if
    \begin{equation}
        \mathcal{T}(\B_{\Phi})-\B_{\Phi}=\Phi.
    \end{equation}
\end{definition}

A key property of boost operators is the following~\cite{pozsgayCurrentOperatorsIntegrable2020,suraceWeakIntegrabilityBreaking2023}.
\begin{prop}\label{prop:boost_inv}
    Let $\Phi\in\localop\cap\transop$ and $\Psi\in\transop$.  
    If $\Psi$ is a conserved quantity of $\Phi$, then $[\B_{\Phi},\Psi]\in\transop$.
\end{prop}
\begin{proof}
Take any $A\in\Zfinsubset$.  
For a sufficiently large $B\supset A$, we have
\begin{align}
    &\tau_A(1)\ad_{\B_{\Phi}}(\Psi)(A)\tau_A(1)^\dagger\nonumber\\
    &=\tau_A(1)[\B_{\Phi}|_B,\Psi|_B]_A\tau_A(1)^\dagger\nonumber\\
    &=\left(\tau_B(1)[\B_{\Phi}|_B,\Psi|_B]\tau_B(1)^\dagger\right)_{A+1}\nonumber\\
    &=\bigl[\tau_B(1)\B_{\Phi}|_B\tau_B(1)^\dagger,\,
    \tau_B(1)\Psi|_B\tau_B(1)^\dagger\bigr]_{A+1}\nonumber\\
    &=\bigl[\mathcal{T}(\B_\Phi)|_{B+1},\,\Psi|_{B+1}\bigr]_{A+1}\nonumber\\
    &=[\B_{\Phi}|_{B+1},\Psi|_{B+1}]_{A+1}
    +[\Phi|_{B+1},\Psi|_{B+1}]_{A+1}\nonumber\\
    &=[\B_{\Phi}|_{B+1},\Psi|_{B+1}]_{A+1}
    =\ad_{\B_{\Phi}}(\Psi)(A+1),
\end{align}
where in the last line we used that $\Psi$ is a conserved quantity of $\Phi$.
\end{proof}

Under these definitions, 
Theorem~\ref{thm:main} can be restated as the following general theorem.
\begin{thm}\label{thm:main_math}
    Let $\Phi\in\formalsa\cap\localop\cap\transop$ be the Hamiltonian of the system, 
    with $\B_{\Phi}$ a boost operator associated with $\Phi$.  
    If $[\B_{\Phi},\Phi]$ is a conserved quantity of $\Phi$,  
    then for any $n\geq1$, 
    \begin{equation}
        \Phi^{(n)}\coloneqq(\ad_{\B_{\Phi}})^{n-1}(\Phi)
    \end{equation}
    is a translationally invariant conserved quantity of $\Phi$.  
    Furthermore, the family $\{\Phi^{(n)}\}_{n=1}^\infty$ consists of mutually commuting operators.
\end{thm}
In the above discussion, 
we do not assume that $\Phi$ is at most $2$-local (which is equivalent to $r(\Phi)\le 2$); 
rather, it is sufficient that $\Phi$ has finite-range interactions.
\change{
For example, consider the translationally and inversion symmetric $S=1/2$ chain with next-nearest-neighbor interactions defined by
\begin{align}
    \hat{H}&=\sum_{j}\hat{h}_j,\\
    \hat{h}_j&=J_{ZZ}^{(2)}Z_{j-1}I_{j}Z_{j+1}
    +J_{XZ}^{(1)}\bigl(X_jZ_{j+1}+Z_{j-1}X_{j}\bigr)
    +h_Y Y_j .
\end{align}
This model is known to be an essentially unique integrable point within the class of translationally and inversion symmetric next-nearest-neighbor $S=1/2$ chains~\cite{shiraishiCompleteClassificationIntegrability2025}.
Define
\begin{equation}
    \hat{B}=\sum_{j}j\,\hat{h}_j .
\end{equation}
Then $\hat{B}$ is a boost operator in the sense of Definition~\ref{def:boost_general}, and one has $[\hat{H},[\hat{B},\hat{H}]]=0$.
Consequently, the integrability of this system also follows from Theorem~\ref{thm:main_math}.
}

\change{On the other hand, boost-generated conservation laws are not universal for $3$-local (or higher) integrable systems; 
see Ref.~\cite{gomborIntegrableSpinChains2021} for a broader integrability test.
}

\section{Proof of Main Theorem}\label{sec:proof}
For notational simplicity, 
we use the same notation as in Theorem~\ref{thm:main} throughout the proof. 
The proof of Theorem~\ref{thm:main_math} proceeds in exactly the same way as that of Theorem~\ref{thm:main}, 
upon replacing $\hat{H}$, $\hat{Q}^{(n+1)}$, and $\hat{B}_{\hat{H}}$ with $\Phi$, $\Phi^{(n)}$, and $\mathfrak{B}_\Phi$, respectively.

If all $\hat{Q}^{(n)}$ are conserved quantities of $\hat{H}$, 
their mutual commutativity follows immediately from the Jacobi identity. 
In fact, the following standard lemma~\cite{deleeuwClassifyingIntegrableSpin12019,zhangBootstrappingRmatrix2026a} holds. 

\begin{lem}\label{lem:property_boost}
    If $[\hat{H},\hat{Q}^{(m+1)}]=0$ holds for all $m\geq 1$, 
    then $[\hat{Q}^{(k+1)},\hat{Q}^{(l+1)}]=0$ holds for $k,l\geq 1$.
\end{lem}
\begin{proof}
    We proceed by induction on $M=k+l(\geq 2)$.
    For $M=2$, it is trivial that $[\hat{H},\hat{H}]=0$ holds.  
    Assume that the claim is valid up to $M=M'-1\geq 2$, and consider $M=M'$.  
    From the Jacobi identity and the induction hypothesis, 
    for $1\leq k\leq M'-1$ with $l=M'-k$, we have
    \begin{align}
    [\hat{Q}^{(k+1)},\hat{Q}^{(l+1)}]
    &=-[\hat{Q}^{(k)},\hat{Q}^{(l+2)}]+[\hat{B}_{\hat{H}},[\hat{Q}^{(k)},\hat{Q}^{(l+1)}]]\nonumber\\
    &=-[\hat{Q}^{(k)},\hat{Q}^{(l+2)}]\nonumber\\
    &=[\hat{Q}^{(k-1)},\hat{Q}^{(l+3)}]-[\hat{B}_{\hat{H}},[\hat{Q}^{(k-1)},\hat{Q}^{(l+2)}]]\nonumber\\
    &=[\hat{Q}^{(k-1)},\hat{Q}^{(l+3)}]\nonumber\\
    &=\dots\nonumber\\
    &=(-1)^{k-1}[\hat{Q}^{(2)},\hat{Q}^{(M')}]=0, 
    \end{align}
    where we used the assumption $[\hat{H},\hat{Q}^{(M')}]=0$ in the last equality. 
    Thus the claim follows.
\end{proof}

From this lemma, it suffices to prove Theorem~\ref{thm:main} by showing: 
If $[\hat{H},\hat{Q}^{(3)}]=0$ holds, 
then $[\hat{H},\hat{Q}^{(n)}]=0$ also holds for all $n\geq4$.  
The key step is the following observation.
\begin{lem}\label{lem:key}
    If $[\hat{H},\hat{Q}^{(3)}]=0$ and $[\hat{H},\hat{A}]=0$ holds for some local operator $\hat{A}$, 
    then $[\hat{H},[\hat{H},[\hat{B}_{\hat{H}},\hat{A}]]]=0$.
\end{lem}
\begin{proof}
    From the Jacobi identity, we have
    \begin{align}
        [\hat{H},[\hat{B}_{\hat{H}},\hat{A}]]
        &=-[\hat{Q}^{(3)},\hat{A}]+[\hat{B}_{\hat{H}},[\hat{H},\hat{A}]]\nonumber\\
        &=-[\hat{Q}^{(3)},\hat{A}].
    \end{align}
    Since both $\hat{Q}^{(3)}$ and $\hat{A}$ commute with $\hat{H}$, we obtain
    \begin{align}
        [\hat{H},[\hat{H},[\hat{B}_{\hat{H}},\hat{A}]]]
        &=-[\hat{H},[\hat{Q}^{(3)},\hat{A}]]=0,
    \end{align}
    which proves the claim.
\end{proof}

In view of this lemma, it remains to establish the following.
\begin{lem}\label{lem:double_commutator}
    For a translationally invariant $n(\geq0)$-local operator $\hat{A}^{(n)}$,  
    if $[\hat{H},[\hat{H},\hat{A}^{(n)}]]=0$, then $[\hat{H},\hat{A}^{(n)}]=0$.
\end{lem}
\begin{proof}
    Consider a finite system of size $N$ with the periodic boundary condition.  
    Note that the commutativity of translationally invariant local operators is a local property, 
    and hence independent of $N$.  
    Thus, for sufficiently large $N$ compared with $n$, the assumption can be replaced with the finite-size relation
    $[\hat{H}_{[N]},[\hat{H}_{[N]},\hat{A}^{(n)}_{[N]}]]=0$ 
    (See Corollary~\ref{cor:locality_of_conservation} in 
    Appendix~\ref{sec:correspondence} for a formal proof).  
    By considering the Hilbert-Schmidt norm of $[\hat{H}_{[N]},\hat{A}^{(n)}_{[N]}]$, 
    we obtain 
    \begin{align}
    \|[\hat{H}_{[N]},\hat{A}^{(n)}_{[N]}]\|_{\mathrm{HS}}^2
    &=\Tr([\hat{H}_{[N]},\hat{A}^{(n)}_{[N]}]^\dagger [\hat{H}_{[N]},\hat{A}^{(n)}_{[N]}])\nonumber\\
    &=\Tr([{\hat{A}^{(n)\dagger}_{[N]}},\hat{H}_{[N]}][\hat{H}_{[N]},\hat{A}^{(n)}_{[N]}])\nonumber\\
    &=\Tr({\hat{A}^{(n)\dagger}_{[N]}}[\hat{H}_{[N]},[\hat{H}_{[N]},\hat{A}^{(n)}_{[N]}]])=0,
    \end{align}
    which implies $[\hat{H}_{[N]},\hat{A}^{(n)}_{[N]}]=0$.
    Therefore, the corresponding infinite-system relation $[\hat{H},\hat{A}^{(n)}]=0$ also holds.
\end{proof}
It should be noted that we cannot dispense with the assumption of translational invariance in Lemma~\ref{lem:double_commutator}, 
since $[\hat{H},[\hat{H},\hat{B}_{\hat{H}}]]=0$ does not necessarily imply that 
$\hat{Q}^{(3)}=-[\hat{H},\hat{B}_{\hat{H}}]=0$. 

We now summarize the proof of Theorem~\ref{thm:main}.
\begin{proof}[Proof of Theorem~\ref{thm:main}]
We prove by induction 
to show that for all $n\geq 2$, 
$\hat{Q}^{(n)}$ is a translationally invariant conserved quantity of $\hat{H}$.  
For $n=2$, the statement is trivial since $\hat{Q}^{(2)}=\hat{H}$.  

Assume the claim holds for $n=m-1\geq 2$, and consider the case $n=m$. 
As stated in Eq.~\eqref{eq:trans_inv_from_boost} (and Proposition~\ref{prop:boost_inv}), 
if $\hat{Q}^{(m-1)}$ is a translationally invariant conserved quantity, 
then $\hat{Q}^{(m)}$ is also translationally invariant.  
From Lemma~\ref{lem:key}, we have 
$[\hat{H},[\hat{H},\hat{Q}^{(m)}]]=0$, 
and by Lemma~\ref{lem:double_commutator}, 
it follows that $[\hat{H},\hat{Q}^{(m)}]=0$. 
This completes the proof.
\end{proof}

\section{Dichotomy on integrability}\label{sec:dichotomy}
\change{For translationally invariant quantum spin chains with finite-range interactions}, 
a generic system that is not integrable is expected to be (completely) nonintegrable, 
meaning that it possesses no local conserved quantities other than trivial ones such as the Hamiltonian and the total magnetization. 
Indeed, this expectation has been rigorously confirmed for spin-$1/2$ chains with parity-symmetric interactions~\cite{yamaguchiCompleteClassificationIntegrability2024,yamaguchiProofAbsenceLocal2024,shiraishiCompleteClassificationIntegrability2025} 
by applying the method developed in Ref.~\cite{shiraishiProofAbsenceLocal2019} to rigorously prove nonintegrability for specific systems. 
However, these results rely on an \textit{ad hoc}, case-by-case verification of nonintegrability. 

In the following, we address how to rigorously establish this belief for a broader class of Hamiltonians \change{with nearest-neighbor interactions}, 
namely the \emph{dichotomy} on integrability: either there exist infinitely many local conserved quantities, or there exist none beyond the trivial ones. 
Specifically, we combine Theorem~\ref{thm:main}, which provides a sufficient condition for integrability, 
with a sufficient condition for nonintegrability 
that the author has recently proved~\cite{hokkyoRigorousTestQuantum2025}.

We assume that the interaction $\hat{h}^{(2)}$ is sufficiently complicated in the sense of Ref.~\cite{hokkyoRigorousTestQuantum2025}. 
First, for any $\hat{a}^{(1)}\in\bound_0$, we assume that
\begin{equation}
  \begin{aligned}
    [\hat{I}\otimes \hat{a}^{(1)},\hat{h}^{(2)}_{1}]=0 &\;\Rightarrow\; \hat{a}^{(1)}=0,\\
    [\hat{a}^{(1)}\otimes \hat{I},\hat{h}^{(2)}_{1}]=0 &\;\Rightarrow\; \hat{a}^{(1)}=0,
  \end{aligned}
  \label{eq:assumption_injective}
\end{equation}
which is a slightly stronger version of the condition that an on-site operator $\hat{a}^{(1)}$ is not a conserved quantity. 

Next, for a family of operators $\hat{a}^{(2)}_i\in\bound_{0i}\otimes\bound_{0i+1}$, we assume that
\begin{equation}
    [\hat{h}_i^{(2)},\hat{a}^{(2)}_{i+1}] = [\hat{a}^{(2)}_i,\hat{h}_{i+1}^{(2)}]\quad(\forall i)
    \;\Rightarrow\;
    \hat{a}^{(2)}_i=\alpha\,\hat{h}_i^{(2)}\quad(\forall i),
    \label{eq:assumption_1dim}
\end{equation}
where $\alpha\in\C$ is a constant independent of $i$.  
This condition implies that the only two-local conserved quantity is the Hamiltonian itself.

\begin{prop}\label{prop:dichotomy}
Let $\mathbb{H}$ be a set of Hamiltonians satisfying assumptions~\eqref{eq:assumption_injective} and \eqref{eq:assumption_1dim}. For any $\hat{H}\in\mathbb{H}$, suppose that the following holds. 
\begin{enumerate}[label=(\Alph*)]
\item \change{
Let $\hat{Q}^{(3)}_{[3]}=-\sum_{i\in\Z}[\hat{h}_i^{(2)},\hat{h}_{i+1}^{(2)}]$ 
be the $3$-local part of the boosted operator (Eq.~\eqref{eq:n-local_boosted}). 
If there exists a two-local operator $\hat{R}^{(2)}$ such that $[\hat{H},\hat{Q}^{(3)}_{[3]}+\hat{R}^{(2)}]$ is at most two-local, 
then there exists 
a boost operator $\hat{B}$ in the sense of Definition~\ref{def:boost_general} such that $\hat{Q}^{(3)}=[\hat{B},\hat{H}]$ satisfies $[\hat{H},\hat{Q}^{(3)}]=0$.
}
\label{cond:dichotomy}
\end{enumerate}
Then exactly one of the following holds. 
\begin{itemize}
    \item $\hat{H}$ is integrable in the sense that there exist mutually commuting $k$-local conserved quantities for all $k\geq 2$.
    \item $\hat{H}$ is nonintegrable in the sense that the only $k$-local conserved quantities are the Hamiltonian itself ($k=2$) and one-local operators.
\end{itemize}
\end{prop}

\begin{proof}
From condition~\ref{cond:dichotomy}, one of the following holds. 
\begin{itemize}
    \item $\hat{H}$ commutes with $\hat{Q}^{(3)}=[\hat{B},\hat{H}]$.
    \item For any two-local operator $\hat{R}^{(2)}$, the commutator $[\hat{H},\hat{Q}^{(3)}_{[3]}+\hat{R}^{(2)}]$ is three-local.
\end{itemize}
In the former case, Theorem~\ref{thm:main} implies the first conclusion.  
In the latter case, the result of Ref.~\cite{hokkyoRigorousTestQuantum2025} implies the second. 
\end{proof}

Therefore, once condition~\ref{cond:dichotomy} is established, the system admits a complete dichotomy: it is either integrable or nonintegrable. Conversely, if a ``partially integrable'' system were to exist, then within the framework of assumptions~\eqref{eq:assumption_injective} and \eqref{eq:assumption_1dim}, it would necessarily provide a counterexample to condition~\ref{cond:dichotomy}.

\section{Generalizations}\label{sec:generalization}
\subsection{Proof of the Grabowski--Mathieu conjecture for general boosts} \label{sec:proof_GM2}
Grabowski and Mathieu conjectured that the existence of a boost operator 
$\hat{B}$ with $[\hat{B},\hat{H}]$ conserved implies integrability of $\hat{H}$~\cite{grabowskiIntegrabilityTestSpin1995}. 
\change{Theorem~\ref{thm:main} established this for the standard boost \( \hat B_{\hat H} \) associated with the Hamiltonian; in Yang--Baxter-solvable models this is the canonical boost.}
The same reasoning applies to general boosts for Hamiltonians satisfying assumptions~\eqref{eq:assumption_injective} and~\eqref{eq:assumption_1dim}. 
Assume that $\hat{B}$ is Hermitian and of the form 
\begin{equation}
    \hat{B}=\sum_{j\in\Z} \hat{b}^{(2)}_{j}+\sum_{j\in\Z}\hat{b}^{(1)}_j,
    \label{eq:general_boost}
\end{equation}
with $\hat{b}^{(2)}_{j}\in\bound_{0j}\otimes\bound_{0j+1}$ and $\hat{b}^{(1)}_j\in\bound_{0j}$. 
That is, $\hat{B}$ is at most two-local. 
Setting $\hat{Q}^{(2)}\coloneqq\hat{H}$ and $\hat{Q}^{(n+1)}\coloneqq[\hat{B},\hat{Q}^{(n)}]$, 
we obtain the following theorem. 
\begin{thm}\label{thm:proof_GM2}
Let $\hat{H}$ satisfy assumptions~\eqref{eq:assumption_injective} and~\eqref{eq:assumption_1dim}. 
If $\hat{Q}^{(3)}=[\hat{B},\hat{H}]$ is a nonzero conserved quantity, 
then for any $n\geq2$, $\hat{Q}^{(n)}$ is a translationally invariant conserved quantity of $\hat{H}$. 
Moreover, the operators in $\{\hat{Q}^{(n)}\}_{n=2}^\infty$ mutually commute. 
In addition, $\hat{B}$ is restricted to the form
\begin{equation}
    \hat{B}\propto \sum_{j\in\Z}j(\hat{h}^{(2)}_{j}+\hat{h}^{(1)}_j)+\hat{R}+a\hat{H},
    \label{eq:form_of_boost}
\end{equation}
with $a\in\R$, where $\hat{R}$ is one-local and satisfies
\begin{equation}
    [\hat{H},\mathcal{T}(\hat{R})-\hat{R}]=0. 
\end{equation}
\end{thm}
\begin{proof}
The proof that $\hat{B}$ must take the form of Eq.~\eqref{eq:form_of_boost} is given in Appendix~\ref{app:proof_GM2}. 
Accepting this fact, we may, without loss of generality, write $\hat{B}=\hat{B}_{\hat{H}}+\hat{R}$. 

Since Lemmas~\ref{lem:property_boost}, \ref{lem:key} (with $\hat{B}_{\hat{H}}$ replaced by $\hat{B}$), and~\ref{lem:double_commutator} remain valid, 
it suffices to show $\hat{Q}^{(n)}$ is translationally invariant. 
If $\hat{Q}^{(k)}$ is translationally invariant, then
\begin{align}
    \mathcal{T}(\hat{Q}^{(k+1)})-\hat{Q}^{(k+1)}
    &= [\mathcal{T}(\hat{B})-\hat{B},\hat{Q}^{(k)}]\nonumber\\
    &= [\hat{H}+\hat{X},\hat{Q}^{(k)}]
    = [\hat{X},\hat{Q}^{(k)}],
\end{align}
where $\hat{X}\coloneqq\mathcal{T}(\hat{R})-\hat{R}$. 
For $k=2$ this vanishes, proving translational invariance of $\hat{Q}^{(3)}$. 
For general $k$, the one-locality of $\hat{X}$ again ensures $[\hat{X},\hat{Q}^{(k)}]=0$ 
(see Proposition~\ref{prop:trans_inv} in Appendix~\ref{app:proof_GM2}), 
and hence $\hat{Q}^{(k+1)}$ is also translationally invariant. 
Thus, all $\hat{Q}^{(n)}$ are translationally invariant, and commutativity follows as in Theorem~\ref{thm:main}. 
\end{proof}

Theorem~\ref{thm:proof_GM2}, 
which does not assume any specific form of $\hat{B}$, 
provides a complete proof of the Grabowski--Mathieu conjecture (Conjecture 2 in Ref.~\cite{grabowskiIntegrabilityTestSpin1995}) for Hamiltonians satisfying assumptions~\eqref{eq:assumption_injective} and~\eqref{eq:assumption_1dim}: 
the existence of a boost generating a single nonzero local conserved quantity implies integrability. 
Eq.~\eqref{eq:form_of_boost} further shows that any boost coincides with the standard $\hat{B}_{\hat{H}}$ up to a one-local term. 

In contrast, in systems not satisfying assumption~\eqref{eq:assumption_1dim}, such as the XY model, there exists a boost operator different from the standard $\hat{B}_{\hat{H}}$, 
which in turn generates two independent families of local conserved quantities~\cite{grabowskiStructureConservationLaws1995}. 
This also follows from the following proposition, which generalizes Theorem~\ref{thm:main}. 
\begin{prop}\label{thm:abstract_boost}
\change{Define \( \hat Q^{(2)}\coloneqq\hat H \) and \( \hat Q^{(n+1)}\coloneqq[\hat B,\hat Q^{(n)}] \).}
Suppose that $\hat{Q}^{(3)}$ is a conserved quantity of $\hat{H}$ and that $\hat{X}=\mathcal{T}(\hat{B})-\hat{B}$ is \change{a translationally invariant hermitian operator}. 
If 
\begin{equation}
    [\hat{X},[\hat{B},\hat{X}]]=0,
    \label{eq:abstract_boost}
\end{equation}
then all operators in $\{\hat{Q}^{(n)}\}_{n=2}^\infty$ are translationally invariant, mutually commuting, and commute with $\hat{X}$. 
\end{prop}
The proof is essentially the same as that of Theorem~\ref{thm:main}, 
except that commutativity with $\hat{X}$ must also be shown (see Appendix~\ref{sec:abstract_boost}).

\subsection{Parametrized Hamiltonians}
Certain integrable models, such as the Hubbard model, have an $R$-matrix that cannot be expressed as a function of a single real variable. 
As a result, they do not satisfy the standard Reshetikhin condition and, 
moreover, they admit no boost operator~\cite{grabowskiStructureConservationLaws1995}. 
Instead, they are associated with families of parametrized 
integrable
Hamiltonians obeying a generalized Reshetikhin condition~\cite{linksLadderOperatorOnedimensional2001}. 
As a counterpart to Theorem~\ref{thm:main}, 
we show that if a parametrized Hamiltonian satisfies this generalized condition, 
then it is integrable.

We assume that the Hamiltonian takes the form
\begin{equation}
    \hat{H}(\lambda)=\sum_{i\in\Z} \hat{h}^{(2)}_{i}(\lambda)+\sum_{i\in\Z}\hat{h}^{(1)}_i(\lambda),
    \label{eq:Hamiltonian_param}
\end{equation}
where $\lambda$ belongs to some open subset $U\subset\R^D$. 
For a smooth function 
$v:U\to\R^D$, define
\begin{align}
    \hat{Q}^{(2)}_v(\lambda)&\coloneqq\hat{H}(\lambda),\\
    \hat{Q}^{(n+1)}_v(\lambda)&\coloneqq[\hat{B}_{\hat{H}(\lambda)},\hat{Q}^{(n)}_v(\lambda)]+(v\cdot\nabla)\hat{Q}^{(n)}_v(\lambda),
\end{align}
where $v\cdot\nabla=\sum_{a=1}^Dv_a\partial_{\lambda_a}$. 
In particular,
\begin{equation}
    \hat{Q}^{(3)}_v(\lambda)=-\sum_{i\in\Z}[\hat{h}_{i}(\lambda),\hat{h}_{i+1}(\lambda)]
    +(v\cdot\nabla)\hat{H}(\lambda),
\end{equation}
where $\hat{h}_i(\lambda)=\hat{h}^{(2)}_{i}(\lambda)+\hat{h}^{(1)}_{i}(\lambda)+\hat{c}^{(1)}_{i}(\lambda)-\hat{c}^{(1)}_{i+1}(\lambda)$
for some translationally invariant $1$-local operator 
$\hat{C}_v(\lambda)=\sum_{i\in\Z}\hat{c}^{(1)}_{v,i}(\lambda)$. 
The generalized Reshetikhin condition requires $[\hat{Q}^{(3)}_v(\lambda),\hat{H}(\lambda)]=0$ for all $\lambda\in U$, or its local form~\cite{linksLadderOperatorOnedimensional2001}. 

In parallel with the standard case (Theorem~\ref{thm:main}), the following holds. 
\begin{thm}\label{thm:generalized}
    If there exists a smooth function $v:U\to\R^D$ such that
    $\hat{Q}^{(3)}_v(\lambda)$ is a conserved quantity of $\hat{H}(\lambda)$ for every $\lambda\in U$, 
    then $\hat{Q}^{(n)}_v(\lambda)$ is a translationally invariant conserved quantity of $\hat{H}(\lambda)$ for all $n\geq2$,
    and the operators in $\{\hat{Q}^{(n)}_v(\lambda)\}_{n=2}^\infty$ mutually commute.
\end{thm}
We note that setting $v\equiv 0$ reproduces Theorem~\ref{thm:main} as a special case, 
and is consistent with the original Reshetikhin condition for $R$-matrices of difference form.

The proof parallels that of Theorem~\ref{thm:main}, 
using the fact that $\mathscr{B}_v\coloneqq[\hat{B}_{\hat{H}(\lambda)},-]+v\cdot\nabla$ satisfies the Jacobi identity:
\begin{equation}
    \mathscr{B}_v\big([\hat{A},\hat{C}]\big)(\lambda)
    =[\mathscr{B}_v(\hat{A})(\lambda),\hat{C}(\lambda)]
    +[\hat{A}(\lambda),\mathscr{B}_v(\hat{C})(\lambda)].
\end{equation}

Although the above theorem focuses on a single function $v$ satisfying the generalized Reshetikhin condition,
it may also happen that $\hat{Q}_w^{(3)}(\lambda)$ is a conserved quantity for another function $w:U\to\R^D$~\cite{takizawaLadderOperatorsIntegrable2003}.  
In such cases, one can similarly show that $\hat{Q}_v^{(n)}(\lambda)$ and $\hat{Q}_w^{(m)}(\lambda)$ commute.
More strongly, the following holds.
\begin{prop}\label{prop:boost_and_Q}
Let $\hat{A}(\lambda)$ be a translationally invariant and local operator smoothly parametrized over $U\subset\R^D$.  
If both $\hat{Q}^{(3)}_v(\lambda)$ and $\hat{A}(\lambda)$ are conserved quantities of $\hat{H}(\lambda)$ for every $\lambda\in U$, then 
\[
[\hat{Q}^{(n)}_v(\lambda),\hat{A}(\lambda)]=0
\]
for all $n\geq 2$. 
\end{prop}
\begin{proof}
We prove a stronger statement:  
$(\mathscr{B}_v)^{n-2}(\hat{A})(\lambda)$ is translationally invariant and
\[
[\hat{Q}^{(n-k)}_v(\lambda),(\mathscr{B}_v)^k(\hat{A})(\lambda)]=0
\]
holds for all $0\leq k\leq n-2$.  
We establish this claim by induction on $n(\geq 2)$.  
Fix $\lambda\in U$ and suppress the argument for brevity.  

For $n=2$, the claim follows from the assumption.  
Suppose the claim holds for $n=n'-1$ with $n'\geq 3$, and prove it for $n=n'$.  

First, we show that $(\mathscr{B}_v)^{n'-2}(\hat{A})$ is translationally invariant.  
Since $(\mathscr{B}_v)^{n'-3}(\hat{A})$ is translationally invariant, we have
\begin{align}
&\mathcal{T}((\mathscr{B}_v)^{n'-2}(\hat{A}))-(\mathscr{B}_v)^{n'-2}(\hat{A})\nonumber\\
&=[\mathcal{T}(\hat{B}_{\hat{H}})-\hat{B}_{\hat{H}},\change{(\mathscr{B}_v)^{n'-3}(\hat{A})}]\nonumber\\
&=[\hat{H},\change{(\mathscr{B}_v)^{n'-3}(\hat{A})}]=0.
\end{align}

It remains to show $[\hat{Q}^{(n'-k)}_v,(\mathscr{B}_v)^k(\hat{A})]=0$
for $0\leq k\leq n'-2$.  
Using the Jacobi identity, we have
\begin{align}
    [\hat{Q}^{(n')}_v,\hat{A}]
    &=\mathscr{B}_v([\hat{Q}^{(n'-1)}_v,\hat{A}])
    -[\hat{Q}^{(n'-1)}_v,\mathscr{B}_v(\hat{A})]\nonumber\\
    &=-[\hat{Q}^{(n'-1)}_v,\mathscr{B}_v(\hat{A})]\nonumber\\
    &=-\mathscr{B}_v([\hat{Q}^{(n'-2)}_v,\mathscr{B}_v(\hat{A})])
    +\change{[\hat{Q}^{(n'-2)}_v,(\mathscr{B}_v)^2(\hat{A})]}\nonumber\\
    &=\change{[\hat{Q}^{(n'-2)}_v,(\mathscr{B}_v)^2(\hat{A})]}\nonumber\\
    &\dots\nonumber\\
    &=(-1)^{n'}[\hat{Q}^{(2)}_v,(\mathscr{B}_v)^{n'-2}(\hat{A})]\nonumber\\
    &=(-1)^{n'}[\hat{H},(\mathscr{B}_v)^{n'-2}(\hat{A})].
\end{align}
In particular, all commutators $[\hat{Q}^{(n'-k)}_v,(\mathscr{B}_v)^k(\hat{A})]$ are proportional to each other.  
Since Theorem~\ref{thm:generalized} guarantees $[\hat{H},\hat{Q}^{(n')}_v]=0$, we obtain
\begin{align}
    [\hat{H},[\hat{H},(\mathscr{B}_v)^{n'-2}(\hat{A})]]
    =(-1)^{n'}[\hat{H}, [\hat{Q}^{(n')}_v,\hat{A}]]=0.
\end{align}
By Lemma~\ref{lem:double_commutator}, this implies 
$[\hat{H},(\mathscr{B}_v)^{n'-2}(\hat{A})]=0$.  
Therefore, we conclude that $[\hat{Q}^{(n'-k)}_v,(\mathscr{B}_v)^k(\hat{A})]=0$ holds for all $0\leq k\leq n'-2$. 
\end{proof}

If $[\hat{H}(\lambda),\hat{Q}_w^{(3)}(\lambda)]=0$,  
then by Theorem~\ref{thm:generalized}, $\hat{Q}_w^{(m)}(\lambda)$ is a conserved quantity.  
Applying Proposition~\ref{prop:boost_and_Q} with $\hat{A}(\lambda)=\hat{Q}_w^{(m)}(\lambda)$ yields
\[
[\hat{Q}^{(n)}_v(\lambda),\hat{Q}_w^{(m)}(\lambda)]=0.
\]

\section{Conclusion and outlook}
In this paper, we have shown that 
a single conservation law generated by the boost operator gives rise to an infinite hierarchy of commuting local charges, or integrability. 
This result resolves a long-standing conjecture on quantum spin chains 
and provides a rigorous foundation for a heuristic sufficient criterion for integrability based on the Reshetikhin condition, independent of any additional algebraic assumptions. 
\change{Furthermore, it clarifies that the standard sequence of local conserved quantities associated with Hamiltonians obtained via the canonical Yang–Baxter construction is already encoded in the lowest nontrivial conservation law}.

\change{
For systems beyond the assumptions~\eqref{eq:assumption_injective} and~\eqref{eq:assumption_1dim},
it remains open 
whether there exist nonstandard boost operators that generate conserved quantities but do not satisfy Eq.~\eqref{eq:abstract_boost}, and, if so, whether they still imply integrability.
Another important open problem is to elucidate the structure of conserved quantities
in Yang--Baxter solvable models for which the Reshetikhin condition is not directly applicable,
including models constructed from staggered transfer matrices~\cite{barievIntegrableSpinChain1991,vegaNewIntegrableQuantum1992,frahmIntegrableModelsCoupled1996,abadIntegrableSu3Spin1997}}.

\change{
Moreover, since our analysis has addressed only local conserved charges,
its connection to the underlying Yang--Baxter data (such as an $R$-matrix or a transfer matrix)
remains to be clarified}.
In particular, whether one can reconstruct an $R$-matrix directly from a single conserved quantity
remains a challenging problem.
This problem is, in fact, equivalent to an old conjecture
known as the tangential star--triangle hypothesis
\cite{jimboRemarksDifferentialApproach1984,idzumiSolvableNineteenvertexModels1994}.

Our results suggest the possibility of a sharp dichotomy between integrability and nonintegrability 
\change{for translationally invariant quantum spin chains with nearest-neighbor interactions}.
Very recently, Shiraishi and Yamaguchi demonstrated condition~\ref{cond:dichotomy} for $SU(2)$-symmetric quantum spin chains, 
and Proposition~\ref{prop:dichotomy} has been employed to establish a dichotomy theorem~\cite{shiraishiDichotomyTheoremSeparating2026}. 
Extending such results to more general classes of models 
will be crucial for understanding both the singular nature of integrability 
and the universality of nonintegrability.

\appendix
\section{Correspondence between finite and infinite systems}\label{sec:correspondence}
We summarize here the correspondence between finite and infinite chains.
Let $N\in\Z_{\geq 1}$.
For integers $a,b\in\mathbb{Z}$ with $a\leq b$, 
we write $\zinterval{a,b}\coloneqq\{c\in\Z\mid a\leq c\leq b\}$

\subsection{Lattice, state space, and operators in finite systems}
We consider a finite quantum chain of length $N$ subject to \change{the periodic boundary condition}. 
The set of sites is represented by $\zinterval{1,N}$.
For $i\in\mathbb{Z}$ we define the projection $\pi(i)\in\zinterval{1,N}$ by
\begin{equation}
    \pi_{[N]}(i)\coloneqq i-N\left(\left\lceil\frac{i}{N}\right\rceil-1\right)\in\zinterval{1,N}.
\end{equation}
Conversely, for $n\in\zinterval{1,N}$ we define a cut of the chain $s_{[N],n}$ by
\begin{equation}
    s_{[N],n}
    :\zinterval{1,N}\ni i
    \mapsto
    \begin{cases}
        i & (i\geq n)\\
        i+N & (i<n)
    \end{cases}
    \ \in\zinterval{n,n+N-1}.
\end{equation}
This serves as a local inverse of $\pi_{[N]}$:
\begin{align}
    \pi_{[N]}\circ s_{[N],n}(i)&=i,\\
    s_{[N],m}\circ \pi_{[N]}(j)&=
    \begin{cases}
        \vdots\\
        j+N, & (m-N\leq j<m),\\
        j, & (m\leq j<m+N),\\
        j-N, & (m+N\leq j<m+2N),\\
        \vdots
    \end{cases} \nonumber\\
    &=\change{j-N\left\lfloor\frac{j-m}{N}\right\rfloor}.
\end{align}
In what follows, 
we omit the explicit dependence on $N$ and 
simply write $\pi$ and $s_n$ for these maps.

Using $s_n$, the range of a subset $A\subset\zinterval{1,N}$ is defined as
\begin{equation}
    r_{[N]}(A)\coloneqq \min_{n\in\zinterval{1,N}}r(s_n(A)) \ \in\zinterval{0,N}.
\end{equation}
In general we have $r_{[N]}(A)\leq r(A)$.
Moreover, if $1\leq r_{[N]}(A)\leq (N+1)/2$, then there exists a unique $a\in A$ such that
\begin{equation}
    \pi(a+r_{[N]}(A)-1)\in A\subset\pi(\zinterval{a,a+r_{[N]}(A)-1}).
\end{equation}
This $a$ is called the \emph{left-most site} of $A$.
In particular, the following property holds:
\begin{equation}
    r(s_a(A))=r_{[N]}(A)\leq\frac{N+1}{2}<\frac{N+3}{2}\leq r(s_{a'}(A))
    \label{eq:property_left-most}
\end{equation}
for $a'\in A\setminus\{a\}$. 

For $i\in\zinterval{1,N}$, $A\subset\zinterval{1,N}$, and $x\in\mathbb{Z}$, 
we define addition modulo $N$ corresponding to translation by
\begin{align}
    i+_N x&\coloneqq \pi(i+x)\in\zinterval{1,N},\\
    A+_N x&\coloneqq \{a+_N x\mid a\in A\}\subset\zinterval{1,N}.
\end{align}

For $A\in\Zfinsubset$ and an injection $f:A\to A'\subset\mathbb{Z}$,
a unitary map $\tau(f):\Hilb^{\otimes A}\to\Hilb^{\otimes f(A)}$ is naturally defined.
In particular, $\tau_{[N]}(x)\coloneqq \tau(\bullet+_N x)$ represents the translation map on the finite chain.
This $\tau$ satisfies $\tau(\id_A)=I_A$ and $\tau(f\circ g)=\tau(f)\tau(g)$, and in particular, for $0\leq x\leq N-1$ we have
\begin{equation}
    \tau_{[N]}(x)=\tau(s_{x+1})^\dagger\tau_{\zinterval{1,N}}(x).
\end{equation}

Operators on this system are elements of $\bound^{\otimes\zinterval{1,N}}$.
For $X\in \bound^{\otimes\zinterval{1,N}}\setminus\{0\}$ we define the range of $X$ by
\begin{equation}
    r_{[N]}(X)\coloneqq\max\{r_{[N]}(A)\mid A\subset\zinterval{1,N},\, X_A\neq0_A\}
    ,
\end{equation}
with the convention $r_{[N]}(0)=0$.
Note that $r_{[N]}([X,Y])\leq \max\{r_{[N]}(X)+r_{[N]}(Y)-1,0\}$.

For $A\subset B\in\Zfinsubset$, $X\in\bound^{\otimes A}$, and an injection $f:B\to B'\subset\mathbb{Z}$, we have
\begin{equation}
    \tau(f)X\tau(f)^\dagger=\tau(f|_A)X\tau(f|_A)^\dagger.
    \label{eq:tau_property_1}
\end{equation}
Moreover, for $Y\in\bound^{\otimes B}$,
\begin{equation}
    \tau(f)Y_A\tau(f)^\dagger=(\tau(f)Y\tau(f)^\dagger)_{f(A)}.
    \label{eq:tau_property_2}
\end{equation}

\subsection{Locality of commutators}
From an operator defined on the infinite chain, we can construct its counterpart on a finite chain.

\begin{definition}
    Let $\Phi\in\transop$.
    Its finite-size version $\Phi_{[N]}\in\bound^{\otimes\zinterval{1,N}}$ is defined by
    \begin{equation}
        \Phi_{[N]}\coloneqq\Phi(\emptyset)+\sum_{x=0}^{N-1}\tau_{[N]}(x)\left[\sum_{1\in A_0\subset\zinterval{1,N}}
        \Phi(A_0)\right]
        \tau_{[N]}(x)^\dagger
        .
    \end{equation}
    This yields a translation-invariant operator on the finite chain with the periodic boundary condition.
\end{definition}
For $A\neq\emptyset$, the component $(\Phi_{[N]})_A$ can also be expressed as
\begin{equation}
    (\Phi_{[N]})_A=\sum_{a\in A}\tau(s_a)^\dagger\Phi(s_a(A))\tau(s_a).
    \label{eq:component_finite_op}
\end{equation}
This follows from the computation below:
\begin{widetext}
   \begin{align}
    \Phi_{[N]}&=\Phi(\emptyset)+\sum_{a=1}^{N}\tau(s_{a})^\dagger
    \tau((\bullet+a-1)|_{\zinterval{1,N}})\left[\sum_{1\in A_0\subset\zinterval{1,N}}\Phi(A_0)\right]
    \tau((\bullet+a-1)|_{\zinterval{1,N}})^\dagger\tau(s_{a})\nonumber\\
    &=\Phi(\emptyset)+\sum_{a=1}^{N}
    \tau(s_{a})^\dagger
    \left[\sum_{1\in A_0\subset\zinterval{1,N}}\Phi(A_0+a-1)\right]\tau(s_{a})\nonumber\\
    &=\Phi(\emptyset)+\sum_{a=1}^{N}
    \tau(s_{a})^\dagger
    \left[\sum_{a\in A'\subset\zinterval{a,a+N-1}}
    \Phi(A')\right]\tau(s_{a})\nonumber\\
    &=\Phi(\emptyset)+\sum_{a=1}^{N}
    \tau(s_{a})^\dagger
    \left[\sum_{a\in A\subset\zinterval{1,N}}\Phi(s_a(A))\right]\tau(s_{a})\nonumber\\
    &=\Phi(\emptyset)+\sum_{\emptyset\neq A\subset\zinterval{1,N}}\sum_{a\in A}\tau(s_a)^\dagger\Phi(s_a(A))\tau(s_a).
\end{align} 
\end{widetext}

In general, one has $r_{[N]}(\Phi_{[N]})\leq r(\Phi)$.  
In particular, if $r(\Phi)\leq (N/2)+1$, 
for $A\subset\zinterval{1,N}$ with $1\leq r_{[N]}(A)\leq (N+1)/2$ we have
\begin{equation}
    (\Phi_{[N]})_A=\tau(s_a)^\dagger\Phi(s_a(A))\tau(s_a),
    \label{eq:left-most_expansion}
\end{equation}
where $a\in A$ is the left-most site of $A$.
This is a consequence of Eq.~\eqref{eq:property_left-most}.

The following theorem states that, for sufficiently large $N$, the map to finite chains preserves the commutator structure.

\begin{thm}\label{thm:locality_of_com}
    Let $\Phi,\Psi\in\localop\cap\transop$ with $r(\Phi),r(\Psi)\leq (N+1)/2$.
    Then
    \begin{equation}
        [\Phi,\Psi]_{[N]}=[\Phi_{[N]},\Psi_{[N]}].
    \end{equation}
\end{thm}

To prove this theorem, 
we show the following two lemmas.
\begin{lem}\label{lem:suppl_1}
    Let $A\in\Zfinsubset$ with $r(A)\leq N$, and write $a=\min A$. Then
    \begin{equation}
        s_{\pi(a)}(\pi(A))=A+\pi(a)-a.
    \end{equation}
\end{lem}
\begin{proof}
    Both $s_{\pi(a)}(\pi(A))$ and $A+\pi(a)-a$ are contained in 
    $\zinterval{\pi(a),\pi(a)+N-1}$, 
    on which $\pi$ is injective.
    Hence it suffices to check
    \begin{equation}
        \pi(s_{\pi(a)}(\pi(A)))=\pi(A+\pi(a)-a),
    \end{equation}
    which equals $\pi(A)$ on both sides.
\end{proof}

\begin{lem}\label{lem:suppl_2}
    Suppose $B,B'\subset\zinterval{1,N}$ satisfy $r_{[N]}(B),r_{[N]}(B')\leq (N+1)/2$ and $B\cap B'\neq\emptyset$.
    Let $b,b'$ denote the left-most sites of $B,B'$, respectively.
    Then one of the following holds:
    \begin{align}
        s_b(B')&=s_{b'}(B')+nN\quad (\exists n\in\{0,1\}),\label{eq:suppl_1}\\
        s_{b'}(B)&=s_b(B)+mN\quad (\exists m\in\{0,1\}).\label{eq:suppl_2}
    \end{align}
\end{lem}
\begin{proof}
The case $b=b'$ is trivial, so assume $b\neq b'$.
First suppose $s_b(b')-b\le \frac{N+1}{2}-1$. 
Set $n\coloneqq \frac{s_b(b')-b'}{N}\in\{0,1\}$.
Then the assumption implies $b\le b'+nN\le b+\frac{N+1}{2}-1$. 
Hence we have 
\begin{align}
    s_{b'}(B')+nN&\subset \zinterval{b'+nN,
    b'+nN+\lfloor\tfrac{N+1}{2}\rfloor-1}\nonumber\\
    &\subset \zinterval{b,b+N-1}.
\end{align}
Since $\pi$ is injective on $\zinterval{b,b+N-1}$, 
it suffices to show
\begin{equation}
    \pi\left(s_b(B')\right)=\pi\left(s_{b'}(B')+nN\right),
\end{equation}
which indeed holds since both sides equal $B'$.
Therefore \eqref{eq:suppl_1} holds.  
Similarly, if $s_{b'}(b)\le b'+\frac{N+1}{2}-1$, then with
$m\coloneqq \frac{s_{b'}(b)-b}{N}\in\{0,1\}$
we obtain \eqref{eq:suppl_2}.

It remains to consider the case where both
$s_b(b')-b>\frac{N+1}{2}-1$ and 
$s_{b'}(b)-b'>\frac{N+1}{2}-1$ hold. 
Since these quantities are integers, we have
$s_b(b')-b,\ s_{b'}(b)-b'\ \ge \left\lfloor\frac{N+1}{2}\right\rfloor \ge \frac{N}{2}$.
Moreover, because $b\neq b'$, the cuts satisfy
\begin{equation}
    s_b(b')+s_{b'}(b)=b+b'+N.
\end{equation}
Consequently, $s_b(b')-b=s_{b'}(b)-b'=\frac{N}{2}$ holds, 
which is equivalent to $|b-b'|=\frac{N}{2}$. 
\change{In particular, \(N\) is even.}
Consider the case $b'=b+\frac{N}{2}$ and take $c\in B\cap B'$. 
Then $b\leq c<b'$. However, we have
\begin{align}
   \frac{N+1}{2}&\geq r_{[N]}(B')=r(s_{b'}(B'))\nonumber\\
   &\geq c+N-b'+1\nonumber\\
   &\geq b+N-b'+1=\frac{N}{2}+1>\frac{N+1}{2},
\end{align}
which is a contradiction. 
The case $b=b'+\frac{N}{2}$ can be treated in the same way.

Therefore either $s_b(b')-b\leq \frac{N+1}{2}-1$ or $s_{b'}(b)-b'\leq \frac{N+1}{2}-1$ hold, 
and the lemma follows.
\end{proof}
\begin{proof}[Proof of Theorem~\ref{thm:locality_of_com}]
    Without loss of generality we may assume $r(\Phi),r(\Psi)\geq1$.
    It suffices to show that $([\Phi,\Psi]_{[N]})_A=[\Phi_{[N]},\Psi_{[N]}]_A$ holds for every $A\subset\zinterval{1,N}$.  
    The case $A=\emptyset$ is trivial since both sides vanish, so we assume $A\neq\emptyset$.

    \paragraph*{Step 1. Expansion of $([\Phi,\Psi]_{[N]})_A$.}
    First, let us expand $([\Phi,\Psi]_{[N]})_A$.  
    We abbreviate the relation $C\triangle C'\subsetneq A\subset C\cup C'$ by $(C,C')\triangleright A$.  
    From Eq.~\eqref{eq:component_finite_op} we obtain
    \begin{align}
    &([\Phi,\Psi]_{[N]})_A
    =\sum_{a\in A}\tau(s_a)^\dagger[\Phi,\Psi](s_a(A))\tau(s_a) \nonumber\\
    &=\sum_{a\in A}\sum_{\substack{C,C'\in\Zfinsubset\\r(C),r(C')\leq \frac{N+1}{2} \\(C,C')\triangleright s_a(A)}}
    \tau(s_a)^\dagger[\Phi(C),\Psi(C')]_{s_a(A)}\tau(s_a).
    \end{align}
    For each pair $(C,C')$ in the sum, let $c_0=\min (C\cup C')$.  
    Since $C\cap C'\neq\emptyset$ and $r(C),r(C')\leq (N+1)/2$, we have
    \begin{equation}
        s_a(A)\subset C\cup C'\subset \zinterval{c_0,c_0+N-1}.
    \end{equation}
    Hence $\pi|_{C\cup C'}$ is injective, and by Eqs.~\eqref{eq:tau_property_1}, \eqref{eq:tau_property_2} we have
    \begin{align}
        &\tau(s_a)^\dagger[\Phi(C),\Psi(C')]_{s_a(A)}\tau(s_a)\nonumber\\
        &=\tau(s_a^{-1}|_{s_a(A)})[\Phi(C),\Psi(C')]_{s_a(A)}\tau(s_a^{-1}|_{s_a(A)})^\dagger\nonumber\\
        &=\tau(\pi|_{s_a(A)})[\Phi(C),\Psi(C')]_{s_a(A)}\tau(\pi|_{s_a(A)})^\dagger\nonumber\\
        &=\tau(\pi|_{C\cup C'})[\Phi(C),\Psi(C')]_{s_a(A)}\tau(\pi|_{C\cup C'})^\dagger\nonumber\\
        &=(\tau(\pi|_{C\cup C'})[\Phi(C),\Psi(C')]\tau(\pi|_{C\cup C'})^\dagger)_A\nonumber\\
        &=([\tau(\pi|_{C})\Phi(C)\tau(\pi|_{C})^\dagger,\tau(\pi|_{C'})\Psi(C')\tau(\pi|_{C'})^\dagger])_A.
    \end{align}
    Moreover, for each $(C,C')$ the element $a\in A$ satisfying $(C,C')\triangleright s_a(A)$ is unique. 
    Indeed, if $a'>a$, then $a+N\in s_{a'}(A)$, but since $a\geq c_0$, one has $a+N>c_0+N-1$ and hence $a+N\notin C\cup C'$.  
    The case $a'<a$ is similar.  
    Explicitly, the unique element is given by $a=\min[(C\cup C')\cap\pi^{-1}(A)]$.
    Therefore, we arrive at
    \begin{widetext}
        \begin{equation}
    ([\Phi,\Psi]_{[N]})_A
    =\sum_{\substack{C,C'\in\Zfinsubset\\r(C),r(C')\leq \frac{N+1}{2} \\\exists a\in A,\ (C,C')\triangleright s_a(A)}}
    ([\tau(\pi|_{C})\Phi(C)\tau(\pi|_{C})^\dagger,\tau(\pi|_{C'})\Psi(C')\tau(\pi|_{C'})^\dagger])_A.
    \label{eq:finite_com_1}
    \end{equation}
    \end{widetext}

    \paragraph*{Step 2. Construction of a bijection.}
    To compare $([\Phi,\Psi]_{[N]})_A$ with $[\Phi_{[N]},\Psi_{[N]}]_A$, we introduce
    \begin{align}
        \mathcal{C}&\coloneqq\{(C,C')\in\Zfinsubset^2\mid r(C),r(C')\leq \tfrac{N+1}{2},\nonumber\\ &\quad\quad\quad\quad\quad(C,C')\triangleright s_a(A)\ \text{for some }a\in A\},\\
        \mathcal{C}_{[N]}&\coloneqq\{(B,B')\in(2^{\zinterval{1,N}})^2\mid \nonumber\\ &\quad\quad\ \ r_{[N]}(B),r_{[N]}(B')\leq \tfrac{N+1}{2},(B,B')\triangleright A\}.
    \end{align}
    We now construct a bijection between these two sets.  

    First, given $(C,C')\in\mathcal{C}$, we show that $(\pi(C),\pi(C'))\in\mathcal{C}_{[N]}$.  
    Clearly $\pi(C)\subset\zinterval{1,N}$.  
    Writing $c=\min C$, Lemma~\ref{lem:suppl_1} yields
    \begin{align}
        r_{[N]}(\pi(C))&\leq r(s_{\pi(c)}(\pi(C)))=r(C+\pi(c)-c)=r(C)\nonumber\\
        &\leq \tfrac{N+1}{2}.
    \end{align}
    Since $\pi$ is injective on $C\cup C'$, we have $\pi(C\cup C')=\pi(C)\cup\pi(C')$ and $\pi(C\cap C')=\pi(C)\cap\pi(C')$.  
    Hence $(\pi(C),\pi(C'))\triangleright \pi(s_a(A))=A$.
    Therefore, $(\pi(C),\pi(C'))\in\mathcal{C}_{[N]}$ holds.  

    Conversely, for $(B,B')\in\mathcal{C}_{[N]}$, define $(C_{B,B'},C'_{B,B'})\in\Zfinsubset^2$ by
    \begin{align}
        C_{B,B'}^{(')}&\coloneqq s_{b_0}(B^{(')})-n_{b_0}N,\\
        n_{b_0}&\coloneqq
        \begin{cases}
            0 & (\min s_{b_0}(A)\leq N),\\
            1 & (N<\min s_{b_0}(A)),
        \end{cases}
    \end{align}
    where, denoting by $b,b'$ the left-most sites of $B,B'$, we set
    \begin{equation}
    b_0=
    \begin{cases}
        b & (s_b(B')=s_{b'}(B')+nN\ \mbox{for some }n\in\{0,1\}),\\
        b' & \mbox{(otherwise)}.
    \end{cases}
    \end{equation}
    By Lemma~\ref{lem:suppl_2}, we have
    \begin{align}
        s_{b_0}(B')&=s_{b'}(B')+nN\quad (\exists n\in\{0,1\}), \label{eq:simultaneous_cut_1}\\
        s_{b_0}(B)&=s_{b}(B)+mN\quad (\exists m\in\{0,1\}). \label{eq:simultaneous_cut_2}
    \end{align}
    
    We now show $(C_{B,B'},C'_{B,B'})\in\mathcal{C}$.  
    From Eqs.~\eqref{eq:simultaneous_cut_1}, \eqref{eq:simultaneous_cut_2} we obtain
    \begin{equation}
        r(C^{(')}_{B,B'})=r(s_{b^{(')}}(B^{(')}))=r_{[N]}(B^{(')})\leq\tfrac{N+1}{2}.
    \end{equation}
    Next, we prove that $(C_{B,B'},C'_{B,B'})\triangleright s_a(A)$ for some $a\in A$.  
    Since $s_{b_0}$ and translations are injective, we have $(C_{B,B'},C'_{B,B'})\triangleright s_{b_0}(A)-n_{b_0}N$.  
    Let $a=\pi(\min s_{b_0}(A))\in A$.  
    It suffices to show $s_{b_0}(A)-n_{b_0}N=s_a(A)$.  
    Take any $a'\in A$. 
    We claim
    \begin{equation}
        a\leq s_{b_0}(a')-n_{b_0}N\leq a+N-1.
    \end{equation}
    Note that $\pi$ is injective on 
    $s_{b_0}(A)\subset\zinterval{b_0,b_0+N-1}$, 
    hence $s_{b_0}(a)=\min s_{b_0}(A)$.  
    If $n_{b_0}=0$, then $s_{b_0}(a)\leq N$ implies $a\geq b_0$.  
    Thus $a\leq s_{b_0}(a')\leq b_0+N-1\leq a+N-1$.
    If $n_{b_0}=1$, then $a<b_0$ and $s_{b_0}(a)=a+N$.  
    Hence
    \begin{align}
    a+N&\leq s_{b_0}(a)=\min s_{b_0}(A)\nonumber\\
    &\leq s_{b_0}(a')\nonumber\\
    &\leq b_0+N-1\leq 2N-1<a+2N-1,
    \end{align}
    which implies $a\leq s_{b_0}(a')-N\leq a+N-1$.  
    Moreover, we have
    \begin{equation}
        \pi\bigl(s_{b_0}(A) - n_{b_0}N\bigr)=A=\pi\bigl(s_a(A)\bigr).
    \end{equation}
    Since $\pi$ is injective on $\zinterval{a,a+N-1}$, 
    we conclude that $s_{b_0}(A)-n_{b_0}N=s_a(A)$, 
    and thus $(C_{B,B'},C'_{B,B'})\in\mathcal{C}$.

    Finally, we show that the two maps are inverses of each other.  
    Clearly $\pi(C^{(')}_{B,B'})=B^{(')}$, so it suffices to show injectivity of $(C,C')\mapsto(\pi(C),\pi(C'))$.  
    Suppose $(\pi(C_1),\pi(C'_1))=(\pi(C_2),\pi(C'_2))=(B,B')$. 
    Note that for $C$ with $r(C)\leq \tfrac{N+1}{2}$, $\pi(\min C)$ is the left-most site of $\pi(C)$.  
    Therefore, if $b$ is the left-most site of $B$, 
    then Lemma~\ref{lem:suppl_1} gives 
    \begin{align}
        C_i&=s_b(B)+n_iN,\label{eq:form_of_Ci}\\
        n_i&=(\min C_i-b)/N\in\Z.
    \end{align}
    Similarly, if $b'$ is the left-most site of $B'$, 
    we have 
    \begin{equation}
        C'_i=s_{b'}(B')+n'_iN.\label{eq:form_of_C'i}
    \end{equation}
    From Eqs.~\eqref{eq:simultaneous_cut_1}, \eqref{eq:simultaneous_cut_2}, 
    the only possible $k\in\Z$ satisfying
    $s_b(B)\cap (s_{b'}(B')+\change{k}N)\neq\emptyset$ is $k=n-m$.  
    Since $C_i\cap C'_i\neq\emptyset$, 
    this forces 
    \begin{equation}
        n'_i=n_i+n-m.
    \end{equation}
    Consequently, $C_2\cup C'_2=(C_1\cup C'_1)+(n_2-n_1)N$.  
    Since $r(C_1\cup C'_1)\leq N$, 
    we must have either $n_1=n_2$ or $(C_1\cup C'_1)\cap (C_2\cup C'_2)=\emptyset$.  
    However, for $a_i\in A$ with $s_{a_i}(A)\subset C_i\cup C'_i$, 
    we have 
    \begin{align}
        \max\{a_1,a_2\}&\in s_{a_1}(A)\cap s_{a_2}(A)\nonumber\\
        &\subset (C_1\cup C'_1)\cap (C_2\cup C'_2).
    \end{align}
    Thus $(C_1\cup C'_1)\cap (C_2\cup C'_2)\neq\emptyset$, 
    forcing $n_1=n_2$.  
    Hence $(C_1,C'_1)=(C_2,C'_2)$, so the map is injective.

    \paragraph*{Step 3. Identification with the commutator on the finite chain.}
    Changing the summation variables in Eq.~\eqref{eq:finite_com_1} via this bijection, we obtain
    \begin{widetext}
     \begin{align}
    ([\Phi,\Psi]_{[N]})_A
    &=\sum_{(B,B')\in\mathcal{C}_{[N]}}
    ([\tau(\pi|_{C_{B,B'}})\Phi(C_{B,B'})\tau(\pi|_{C_{B,B'}})^\dagger,\,
      \tau(\pi|_{C'_{B,B'}})\Psi(C'_{B,B'})\tau(\pi|_{C'_{B,B'}})^\dagger])_A.
    \end{align}    
    By construction, for some $k\in\Z$ one has $C_{B,B'}=s_b(B)+kN$, 
    and hence
    \begin{align}
    \tau(\pi|_{C_{B,B'}})\Phi(C_{B,B'})\tau(\pi|_{C_{B,B'}})^\dagger
    &=\tau(\pi|_{s_b(B)+kN})\tau_{s_b(B)}(kN)\Phi(s_b(B))\tau_{s_b(B)}(kN)^\dagger\tau(\pi|_{s_b(B)+kN})^\dagger\nonumber\\
    &=\tau(\pi|_{s_b(B)})\Phi(s_b(B))\tau(\pi|_{s_b(B)})^\dagger.
    \end{align}
    Here, in the first equality we used the translation invariance of $\Phi$, and in the second equality we applied the following identity:
    \begin{equation}
        \tau(\pi|_{A+kN})\tau_A(kN)=\tau(\pi|_{A})\quad (k\in\Z).
    \end{equation}
    Since the same reasoning applies to $C'_{B,B'}$, we obtain 
    \begin{align}
    ([\Phi,\Psi]_{[N]})_A
    &=\sum_{(B,B')\in\mathcal{C}_{[N]}}
    ([\tau(\pi|_{s_{b}(B)})\Phi(s_{b}(B))\tau(\pi|_{s_{b}(B)})^\dagger,\tau(\pi|_{s_{b'}(B')})\Psi(s_{b'}(B'))\tau(\pi|_{s_{b'}(B')})^\dagger])_A\nonumber\\
     &=\sum_{(B,B')\in\mathcal{C}_{[N]}}
    ([\tau(s_b)^\dagger\Phi(s_{b}(B))\tau(s_b),\tau(s_{b'})^\dagger\Psi(s_{b'}(B'))\tau(s_{b'})])_A\nonumber\\
    &=\sum_{(B,B')\in\mathcal{C}_{[N]}}
    ([(\Phi_{[N]})_B,(\Psi_{[N]})_{B'}])_A\nonumber\\
    &=\sum_{B,B'\subset\zinterval{1,N}}
    ([(\Phi_{[N]})_B,(\Psi_{[N]})_{B'}])_A=[\Phi_{[N]},\Psi_{[N]}]_A,
    \end{align} 
    \end{widetext}
    which establishes the desired identity.
    Here we used Eq.~\eqref{eq:left-most_expansion} in the third equality.  
\end{proof}

Using Theorem~\ref{thm:locality_of_com}, 
we can pass back and forth between the commutativity of operators in the infinite and finite systems.

\begin{cor}\label{cor:locality_of_conservation}
    Let $\Phi,\Psi\in\localop\cap\transop$ with $r(\Phi),r(\Psi)\leq (N+1)/2$.
    Then $[\Phi,\Psi]=0$ implies $[\Phi_{[N]},\Psi_{[N]}]=0$.  
    Conversely, if in addition $r(\Phi)+r(\Psi)\leq (N/2)+2$, 
    then $[\Phi_{[N]},\Psi_{[N]}]=0$ implies $[\Phi,\Psi]=0$.
\end{cor}
\begin{proof}
    The first part follows immediately from Theorem~\ref{thm:locality_of_com}, since $[\Phi,\Psi]=0$ implies $[\Phi,\Psi]_{[N]}=0$.  
    To prove the second part, assume that $r(\Phi)+r(\Psi)\leq \tfrac{N}{2}+2$ and that $[\Phi_{[N]},\Psi_{[N]}]=0$.  
    In this case we have $r([\Phi,\Psi])\leq \tfrac{N}{2}+1$, and thus it suffices to show that $[\Phi,\Psi](A)=0$ for every $A\in\Zfinsubset$ with $r(A)\leq \tfrac{N}{2}+1$.  
    By Theorem~\ref{thm:locality_of_com}, $[\Phi,\Psi]_{[N]}=0$.  
    From Eq.~\eqref{eq:left-most_expansion} together with Lemma~\ref{lem:suppl_1}, if $r(A)\leq \tfrac{N+1}{2}$, then setting $a=\min A$ we obtain
    \begin{widetext}
    \begin{align}
        [\Phi,\Psi](A)
        &=\tau_A(\pi(a)-a)^\dagger[\Phi,\Psi](A+\pi(a)-a)\tau_A(\pi(a)-a)\nonumber\\
        &=\tau_A(\pi(a)-a)^\dagger\tau(s_{\pi(a)})([\Phi,\Psi]_{[N]})_{\pi(A)}\tau(s_{\pi(a)})^\dagger\tau_A(\pi(a)-a)=0.
    \end{align}   
    \end{widetext}

    We now assume that $N$ is even and consider the case $r(A)=\tfrac{N}{2}+1$.  
    Without loss of generality we may take $r(\Phi)+r(\Psi)=\tfrac{N}{2}+2$.  
    In this situation $r(\Phi),r(\Psi)\geq2$.  
    If $r(C)\leq r(\Phi)$ and $r(C')\leq r(\Psi)$ with $C\triangle C'\subsetneq A \subset C\cup C'$, then it must be the case that $r(C)=r(\Phi)$, $r(C')=r(\Psi)$, $A=C\cup C'$, and either $\max C=\min C'$ or $\min C=\max C'$.  
    In either situation we have $|A|\geq3$.
    \change{Hence, if \(|A|\le 2\), there is no pair \((C,C')\) contributing to
\([\Phi,\Psi](A)\), and therefore \([\Phi,\Psi](A)=0\).
Thus it remains to consider the case \(|A|\ge 3\).}

    We claim that if $r(A)=\tfrac{N}{2}+1$ and $|A|\geq3$, 
    then $r(s_{a'}(\pi(A)))>\tfrac{N}{2}+1$ for every $a'\in\pi(A)\setminus\{\pi(\min A)\}$.
    \change{
    By translation invariance, we may assume
    \(\min A=1\) and \(\max A=\frac N2+1\).
    Since \(|A|\ge 3\), there exists \(a_*\in A\) such that
    \(1<a_*<\frac N2+1\).}
    If $2\leq a'\leq a_*$, then $s_{a'}(\pi(A))\supset\{a_*,s_{a'}(1)=N+1\}$, and hence
    \begin{equation}
        r(s_{a'}(\pi(A)))\geq N+2-a_*> \tfrac{N}{2}+1.
    \end{equation}
    If $a_*<a'\leq \tfrac{N}{2}+1$, then $s_{a'}(\pi(A))\supset\{\tfrac{N}{2}+1,s_{a'}(a_*)=N+a_*\}$, and therefore
    \begin{equation}
        r(s_{a'}(\pi(A)))\geq \tfrac{N}{2}+a_*> \tfrac{N}{2}+1.
    \end{equation}

    It then follows from $r([\Phi,\Psi])\leq \tfrac{N}{2}+1$ and Eq.~\eqref{eq:component_finite_op} that, for $a=\min A$,
    \begin{equation}
        (\change{[\Phi,\Psi]}_{[N]})_{\pi(A)}
        =\tau(s_{\pi(a)})^\dagger\change{[\Phi,\Psi]}(s_{\pi(a)}(\pi(A)))\tau(s_{\pi(a)}).
    \end{equation}
    As in the case $r(A)\leq \tfrac{N+1}{2}$, we conclude that $[\Phi,\Psi](A)=0$.
\end{proof}

\section{Additional proofs}
\subsection{Complete proof of Theorem~\ref{thm:proof_GM2}}\label{app:proof_GM2}
We begin with the following lemma.
\begin{lem}\label{lem:one-local_double_commutator}
    Suppose that the Hamiltonian $\hat{H}$ satisfies Assumption~\eqref{eq:assumption_1dim}, 
    and let $\hat{A}$ be a one-local operator. 
    If $[\hat{H},[\hat{H},\hat{A}]]=0$, then $[\hat{H},\hat{A}]=0$.
\end{lem}
\begin{proof}
    One can write 
    \[
        [\hat{H},\hat{A}]=\sum_{j}[\hat{h}_j^{(2)},\hat{a}_j^{(1)}+\hat{a}_{j+1}^{(1)}]+[\hat{h}_j^{(1)},\hat{a}_j^{(1)}],
    \]
    which is at most two-local.  
    From $[\hat{H},[\hat{H},\hat{A}]]=0$ and Assumption~\eqref{eq:assumption_1dim}, it follows that
    \begin{equation}
        [\hat{h}_j^{(2)},\hat{a}_j^{(1)}+\hat{a}_{j+1}^{(1)}]=\alpha\hat{h}_j^{(2)},
        \label{eq:one-local_double_com}
    \end{equation}
    for some site-independent constant $\alpha$.  
    However, the left-hand side of Eq.~\eqref{eq:one-local_double_com} is orthogonal to $\hat{h}_j^{(2)}$ with respect to the Hilbert--Schmidt inner product, so $\alpha=0$ 
    (or $\hat{h}_j^{(2)}=0$).  
    Therefore,
    \[
        [\hat{H},\hat{A}]=\sum_{j}[\hat{h}_j^{(1)},\hat{a}_j^{(1)}],
    \]
    which is in particular one-local.  
    Looking at the one-local component of $[\hat{H},[\hat{H},\hat{A}]]$, we find
    \begin{equation}
        [\hat{h}_j^{(1)},[\hat{h}_j^{(1)},\hat{a}_j^{(1)}]]=0.
    \end{equation}
    Since $\hat{h}_j^{(1)}$ is Hermitian, this implies $[\hat{h}_j^{(1)},\hat{a}_j^{(1)}]=0$.  
    Hence, $[\hat{H},\hat{A}]=0$.
\end{proof}

This lemma is analogous to Lemma~\ref{lem:double_commutator}, except that it relies on one-locality and holds even if $\hat{A}$ is not translation-invariant.  
Using it, we now show that the boost operator $\hat{B}$ must take the form of Eq.~\eqref{eq:form_of_boost}.
\begin{proof}[Proof of Eq.~\eqref{eq:form_of_boost}]
We first determine $\hat{b}_j^{(2)}$.  
Among the components of $\hat{Q}^{(3)}=[\hat{B},\hat{H}]$, the $3$-local part acting nontrivially on sites $j$ and $j+2$ can be written as
\begin{equation}
    [\hat{b}_j^{(2)},\hat{h}_{j+1}^{(2)}]-[\hat{h}_j^{(2)},\hat{b}_{j+1}^{(2)}].
    \label{eq:boosted_three-local}
\end{equation}
From Assumptions~\eqref{eq:assumption_injective} and \eqref{eq:assumption_1dim},  
Eq.~\eqref{eq:boosted_three-local} must be proportional to $[\hat{h}_j^{(2)},\hat{h}_{j+1}^{(2)}]$~\cite{hokkyoRigorousTestQuantum2025}, i.e.,
\begin{equation}
    [\hat{b}_j^{(2)},\hat{h}_{j+1}^{(2)}]-[\hat{h}_j^{(2)},\hat{b}_{j+1}^{(2)}]
    =-\gamma[\hat{h}_j^{(2)},\hat{h}_{j+1}^{(2)}]
\end{equation}
for some $\gamma\in\C$.  
Thus, defining $\hat{a}_j^{(2)}=\hat{b}_j^{(2)}-\gamma j\hat{h}_j^{(2)}$,  
Assumption~\eqref{eq:assumption_1dim} implies $\hat{b}_j^{(2)}=(\gamma j+\alpha)\hat{h}_j^{(2)}$ for some $\alpha\in\C$.  
Hermiticity requires $\gamma,\alpha\in\R$.

From the above, one finds that for some one-local operator $\hat{Y}$,
\[
    \mathcal{T}(\hat{B})-\hat{B}=\gamma \sum_{j\in\Z}\hat{h}_j^{(2)}+\hat{Y}.
\]
Setting $\hat{X}'=\hat{Y}-\gamma\sum_{j\in\Z}\hat{h}_j^{(1)}$, we obtain
\[
    \mathcal{T}(\hat{B})-\hat{B}=\gamma \hat{H}+\hat{X}'.
\]
Since $[\hat{H},[\hat{B},\hat{H}]]=0$, we have
\begin{align}
    0&=[\hat{H},[\mathcal{T}(\hat{B})-\hat{B},\hat{H}]]\nonumber\\
    &=-[\hat{H},[\hat{H},\hat{X}']].
\end{align}
Therefore, Lemma~\ref{lem:one-local_double_commutator} yields $[\hat{H},\hat{X}']=0$.  
Hence,
\begin{align}
    &\hat{B}=\gamma\sum_{j\in\Z}j(\hat{h}^{(2)}_{j}+\hat{h}^{(1)}_j)+\hat{R}'+\alpha\hat{H},\\
    &[\hat{H},\mathcal{T}(\hat{R}')-\hat{R}']=0.
\end{align}
It remains to show that $\gamma\neq0$.

If $\gamma=0$, then $\hat{Q}^{(3)}=[\hat{R}',\hat{H}]$ and $[\hat{H},[\hat{R}',\hat{H}]]=0$.  
By Lemma~\ref{lem:one-local_double_commutator}, this would imply $\hat{Q}^{(3)}=0$, contradicting the assumption that $\hat{Q}^{(3)}\neq0$.  
Therefore, $\gamma\neq0$, and Eq.~\eqref{eq:form_of_boost} follows.
\end{proof}

Finally, we establish the following fact, whose proof had been deferred in Section~\ref{sec:proof_GM2}.
\begin{prop}\label{prop:trans_inv}
    Suppose that the Hamiltonian $\hat{H}$ satisfies Assumption~\eqref{eq:assumption_1dim}.  
    Let $\hat{X}$ be a one-local operator such that
    $\mathcal{T}(\hat{B})-\hat{B}=\hat{H}+\hat{X}$.  
    If $[\hat{H},\hat{X}]=[\hat{H},\hat{Q}^{(3)}]=0$, then
    $[\hat{X},\hat{Q}^{(n)}]=0$ for all $n\geq 2$. 
\end{prop}

To prove this, we use the following lemma.
\begin{lem}\label{lem:one-local_conserve}
    Let $\hat{A}$ be a one-local operator.  
    If $[\hat{H},\hat{A}]=0$, then $[\hat{B}_{\hat{H}},\hat{A}]$ is one-local.
\end{lem}
\begin{proof}
    Write $\hat{A}=\sum_{j\in\Z}\hat{a}_j^{(1)}$.  
    Since $[\hat{H},\hat{A}]=0$, we have $[\hat{h}_j^{(2)},\hat{a}_j^{(1)}+\hat{a}_{j+1}^{(1)}]=0$.  
    On the other hand, the two-local component of $[\hat{B}_{\hat{H}},\hat{A}]$ is 
    $\sum_{j\in\Z}[j\hat{h}_j^{(2)},\hat{a}_j^{(1)}+\hat{a}_{j+1}^{(1)}]$, which vanishes.
\end{proof}

\begin{proof}[Proof of Proposition~\ref{prop:trans_inv}]
    We define $\hat{X}^{(1)}=\hat{X}$ and $\hat{X}^{(k+1)}=[\hat{B},\hat{X}^{(k)}]$.  
    We first show that each $\hat{X}^{(n)}$ is one-local and commutes with $\hat{H}$.  
    For $n=1$, this holds by assumption.  
    Suppose it holds for $n=k$.  
    Then, by Lemma~\ref{lem:one-local_conserve}, $\hat{X}^{(k+1)}$ is also one-local.  
    Moreover, by the same reasoning as in Lemma~\ref{lem:key}, $[\hat{H},[\hat{H},\hat{X}^{(k+1)}]]=0$.  
    Thus, Lemma~\ref{lem:one-local_double_commutator} implies $[\hat{H},\hat{X}^{(k+1)}]=0$.

    Next, we prove $[\hat{X},\hat{Q}^{(n)}]=0$.  
    More generally, we show $[\hat{X}^{(l)},\hat{Q}^{(n)}]=0$ for all $l\geq 1$.  
    For $n=2$, this holds by the above argument.  
    Suppose it holds for $n=k$.  
    Then, by the Jacobi identity,
    \begin{equation}
        [\hat{X}^{(l)},\hat{Q}^{(k+1)}]=
        -[\hat{X}^{(l+1)},\hat{Q}^{(k)}]+[\hat{B},[\hat{X}^{(l)},\hat{Q}^{(k)}]]=0.
    \end{equation}
    The claim follows.
\end{proof}

\subsection{Proof of Proposition~\ref{thm:abstract_boost}}\label{sec:abstract_boost}
We first show that $[\hat{H},\hat{X}]=0$.  
Since $\hat{H}$ is translation-invariant, $[\hat{H},\mathcal{T}(\hat{Q}^{(3)})]=0$ holds as well.  
\change{Using \(\mathcal{T}(\hat{B})=\hat{B}+\hat{X}\)},
this is equivalent to
\begin{equation}
    [\hat{H},\hat{Q}^{(3)}+[\hat{X},\hat{H}]]=0.
\end{equation}
Hence $[\hat{H},[\hat{H},\hat{X}]]=0$, and by Lemma~\ref{lem:double_commutator}, $[\hat{H},\hat{X}]=0$.

Next, we show that $[\hat{X},\hat{Q}^{(3)}]=0$.  
From $\mathcal{T}(\hat{B})-\hat{B}=\hat{X}$, it follows that for any translation-invariant operator $\hat{A}$ commuting with $\hat{X}$, the commutator $[\hat{B},\hat{A}]$ is also translation-invariant.  
In particular, $\hat{Q}^{(3)}=[\hat{B},\hat{H}]$ is translation-invariant.  
Moreover, as in Lemma~\ref{lem:key}, one finds $[\hat{X},[\hat{X},\hat{Q}^{(3)}]]=0$ 
\change{by using Assumption~\eqref{eq:abstract_boost}}.  
Therefore, by Lemma~\ref{lem:double_commutator}, we conclude $[\hat{X},\hat{Q}^{(3)}]=0$.

It then remains to show that $\hat{Q}^{(n)}$ commutes with both $\hat{H}$ and $\hat{X}$.  
This reduces to the following generalization of Lemma~\ref{lem:key}:  
if $[\hat{H},\hat{A}]=[\hat{X},\hat{A}]=0$, then 
\[
    [\hat{H},[\hat{B},\hat{A}]]=[\hat{X},[\hat{B},\hat{A}]]=0.
\]
Indeed, from the Jacobi identity and the assumptions, one has $[\hat{H},[\hat{B},\hat{A}]]=[\hat{A},\hat{Q}^{(3)}]$.  
Since both $\hat{A}$ and $\hat{Q}^{(3)}$ commute with $\hat{H}$, it follows that $[\hat{H},[\hat{H},[\hat{B},\hat{A}]]]=0$.  
By Lemma~\ref{lem:double_commutator}, this implies $[\hat{H},[\hat{B},\hat{A}]]=0$.  
The same reasoning shows $[\hat{X},[\hat{B},\hat{A}]]=0$.

Finally, the mutual commutativity of $\{\hat{Q}^{(n)}\}$ can be established in exactly the same way as in Lemma~\ref{lem:property_boost}. \qed

\acknowledgments
The author is very grateful to Naoto Shiraishi, Yuuya Chiba and Mizuki Yamaguchi for helpful discussions. 
The author also thanks Masahito Ueda, Balázs Pozsgay, Chiara Paletta and Jon Links for their valuable comments on the manuscript.
This work was supported by KAKENHI Grant
No. JP25KJ0833 from the Japan Society for the
Promotion of Science (JSPS) and FoPM, a WINGS Program, 
the University of Tokyo.
The author also acknowledges support from JSR Fellowship, the University of Tokyo. 

\bibliography{
bibliography
}

\end{document}